%% file: globalweaknullpostposted2.tex
\documentclass{amsart}
\usepackage{latexsym}
\usepackage{linearb}
\usepackage{graphicx}
\usepackage{marvosym}
\usepackage{color}
\makeindex

\newtheorem{remark}{Remark}[section]

\newtheorem{theorem}{Theorem}
\newtheorem*{litthe}{Theorem}
\newtheorem{proposition}{Proposition}[section]
\newtheorem*{conji}{Conjecture}
\newtheorem{conjecture}{Conjecture}

\newtheorem*{claim}{Claim}

\newtheorem{corollary}{Corollary}[section]

\title{Black holes without spacelike singularities}
\author{Mihalis Dafermos}
\thanks{University of Cambridge,
Department of Pure Mathematics and Mathematical Statistics,
Wilberforce Road, Cambridge CB3 0WB United Kingdom}
\begin{document}
\maketitle
\begin{abstract}
It is shown that for  small, spherically symmetric 
perturbations
of asymptotically flat two-ended Reissner--Nordstr\"om data  for the
Einstein--Maxwell--real scalar field system,  the boundary
of the dynamic spacetime which evolves is \emph{globally} represented 
by a bifurcate null hypersurface across which the metric extends continuously.
Under additional assumptions, it is shown that the Hawking mass blows
up identically along this bifurcate null hypersurface, and thus the metric
cannot be extended twice differentiably, in fact, cannot be extended
in a weaker sense characterized at the level of the Christoffel symbols.
The proof combines estimates obtained in previous work with 
an elementary Cauchy stability argument.
There are no restrictions on the size of the support of the scalar field, and the result applies
to both the future and past boundary of spacetime.
In particular, it follows that for an open set in the moduli space of solutions
around Reissner--Nordstr\"om, \emph{there is no spacelike component
of either the future or the past singularity}. 
\end{abstract}

\section{Introduction}
It is well known that for the \emph{Kerr} and \emph{Reissner--Nordstr\"om} families
of spacetimes, then unless the rotation
or charge  parameter vanishes,  the maximal
globally hyperbolic region $(\mathcal{M},g)$ determined 
(as a solution to the Cauchy problem for the Einstein or Einstein--Maxwell system, respectively)  by initial data
is smoothly extendible
to a larger spacetime $(\widetilde{\mathcal{M}},\widetilde{g})$
so that \emph{all} future incomplete causal geodesics $\gamma$ in $\mathcal{M}$ pass into the
extension $\widetilde{\mathcal{M}}\setminus \mathcal{M}$. The boundary
of $\mathcal{M}$
in $\widetilde{\mathcal{M}}$ defines what is known as a \emph{Cauchy horizon}
(denoted $\mathcal{CH}^+$)
and is in fact a bifurcate null hypersurface.
In the case of Reissner--Nordstr\"om, this situation can be illustrated formally
in the well-known Penrose
diagrammatic notation:
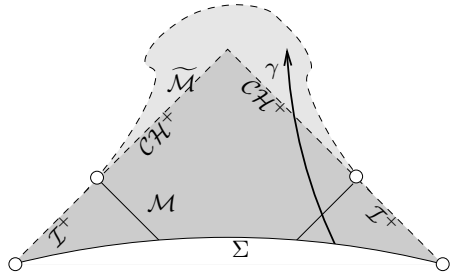
\[
\input{glob_RN.pstex_t}
\]

The above situation may be considered to be
undesirable from the physical or
epistemological point of view as it means that
the fate of classical observers is not determined by what  (from local considerations)
would appear to be an impeccable
initial hypersurface $\Sigma$. The question will not easily go away in view of a celebrated result of Penrose~\cite{Penr} which implies in particular that if one perturbs
the initial data on $\Sigma$, there will still always exist observers who, like
$\gamma$ above, have incomplete worldlines.
Thus, addressing the fate of  such observers who live only for a
finite time is an unavoidable issue for the theory.

From the above point of view, the situation in the celebrated \emph{Schwarzschild}
solution (thought of as a special case of the Reissner--Nordstr\"om/Kerr family corresponding
to vanishing charge/rotation) can be considered to be preferable.
Recall the Penrose diagramme:
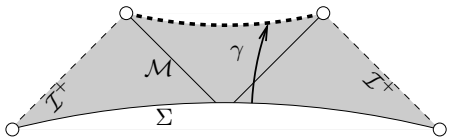
\[
\input{glob_SCH.pstex_t}
\]
Here the curvature blows up along {\bf all} incomplete inextendible causal 
geodesics $\gamma$ as the supremum of affine time is approached,
and, in fact, one cannot extend nontrivially $\mathcal{M}$ to an
$(\widetilde{\mathcal{M}},\widetilde{g})$
such that $\widetilde{g}$ is a $C^0$ metric.
This means that those macroscopic observers who  live only for
 finite time do not just encounter infinite curvature but are pulled apart
by an infinite tidal deformation.
Though perhaps not so pleasant for those observers,
this situation gives the classical theory an attractive sense of ``closure''.
Macroscopic observers either live forever or are destroyed (as classically well-defined objects);
no one is left unaccounted for!

\subsection{Strong cosmic censorship}
Since the Schwarzschild metrics constitute a one-parameter subfamily
of Reissner--Nordstr\"om or Kerr, then in
the class of explicit stationary solutions, it is  the  (undesirable) Reissner--Nordstr\"om/Kerr 
behaviour  that would appear to be generic. 
It turns out, however, that the Reissner--Nordstr\"om 
Cauchy horizon $\mathcal{CH}^+$ is a surface
of infinite \emph{blueshift}, and an old heuristic argument due to Simpson and Penrose~\cite{simp}
suggested that this would cause it  to be unstable
to \emph{dynamic} perturbations of the metric, which would be amplified by 
a blueshift effect.
This motivated the so-called strong\footnote{The word ``strong'' here
is to differentiate between the so-called ``weak'' cosmic censorship, also
due to Penrose, which states that for
generic asymptotically flat data, future null infinity is complete (see~\cite{sings}
for a general formulation, and~\cite{naksinginst} for the seminal proof 
of this conjecture for the
Einstein--scalar field model under spherical symmetry). 
This conjecture has the loose  interpretation
that inextendible, incomplete geodesics must enter black hole regions. In the
formulations given,  
neither conjecture implies the other, but the Reissner--Nordstr\"om/Kerr
behaviour, were it generic, would violate strong cosmic censorship, but
not weak. The traditional epithets reflect precisely this relation.}
 cosmic censorship conjecture:

\begin{conji}[Strong cosmic censorship]
For generic asymptotically flat initial data for ``reasonable'' Einstein--matter systems, 
the maximal Cauchy
development is future  inextendible as a suitably regular Lorentzian manifold.
\end{conji}

The word ``reasonable'' above is meant to exclude in particular matter models for which 
the analogue of the above is false even when one does not couple to gravity\footnote{Here
again, one must be careful as to which notion of maximal development (i.e.~based
on which regularity class) one considers.
For the most classical of all matter models, the 
case of perfect fluids, for instance, the analogue of the above conjecture is false
for the smooth maximal development (in view of the well known formation of shocks),
but may be true for a notion of maximal development based on a wider class of 
solutions admitting shocks and their interactions. In the $1+1$ dimensional case, such 
a class can be built based on the space $BV$, at least for sufficiently small solutions.
See the discussion in Kommemi~\cite{kommemi}.}.
 In particular, the conjecture is meant to apply for the vacuum equations, where
 no matter is present.
   We shall comment on various formulations of `suitably regular' further down.

According to the above conjecture, the Schwarzschild inextendibility property, though special in the class of explicit, stationary solutions, 
would turn out to be generic after all in the class of dynamical solutions. 
In particular, for this conjecture to be true, the internal structure  of
the black hole region of the Reissner--Nordstr\"om/Kerr family would have to be unstable 
to perturbation of initial data.

Since at the time of the formulation of the above 
conjecture, Schwarzschild was the main example  exhibiting
the phenomenon postulated therein as generic,
it naturally served as  a basic model for 
the expected picture. In particular, it was expected that
generically, the singular boundary of spacetime should be 
spacelike, meaning that typically, observers ``falling into the singularity'' eventually
become causally independent, and, moreover, the singularity should be strong, meaning that the
volume element (with respect to a suitable foliation) vanishes and
macroscopic observers falling into the singularity are ``crushed'' (cf.~\cite{tipler}). 
Locally, it is well known that these properties are similar
to the ``big bang'' type singularities of model cosmological spacetimes like  FRW
or the Bianchi class.  (Indeed, the black hole interior of Schwarzschild may
be considered a homogeneous ``cosmological'' solution with spatial
topology $\mathbb R\times \mathbb S^2$ or quotients thereof.) 
Motivated primarily by the cosmological case, a much cited approach
to uncover the form of the  ``generic singularity'' attempted to infer
(via a series of approximations) the existence of an
infinite dimensional family of
spacetimes of prescribed singular behaviour with free functions, 
naively identifying the latter with the degrees of freedom of 
initial data~\cite{BKL2}. This approach was extremely influential in shaping a
`standard' picture of the interiors of generic
black holes, represented by the Schwarzschild Penrose diagramme,
but where the singularity exhibits behaviour as prescribed in~\cite{BKL2}. This
is sometimes known as the BKL scenario.

\subsection{The mass inflation scenario and weak null singularities}
This neat `standard picture' of spacetime singularities fell apart (at least
as a universal prescription)
thanks to the careful study of what
turned out to be a very fruitful class of toy models,
namely spherically symmetric models of collapse in the presence of charge.
These
are the simplest models which on the one hand allow
for the formation of Cauchy horizons (i.e.~they admit Reissner--Nordstr\"om
as a particular solution), while on the other hand, provide
a dynamic degree of freedom on which the blue-shift mechanism can operate
and affect a backreaction on the metric.

\subsubsection{Null dust models}
\label{nulldustmodels}
There is in fact a hierarchy of 
progressively more complicated spherically symmetric charged
models one can consider.
By far the simplest model
quantifying backreaction is that of a self-gravitating ingoing
null dust.  Here the spacetime can be described by an explicit solution,
the \emph{Reissner--Nordstr\"om--Vaidya} metric. This was
studied already in 1981 by Hiscock~\cite{hiscock}.  These solutions
admit a null boundary across which the spacetime
is extendible continuously. All curvature scalars remain bounded
on the boundary, suggesting at first instance that
the spacetime is regular there, just as for the Reissner--Nordstr\"om Cauchy
horizon;
more careful analysis quickly reveals
that various components blow up with respect
to a freely falling
frame. Thus, continuous spacetime extensions beyond the boundary
fail to have $C^2$ metric,
and the boundary can be viewed
as a null singularity, whose singular nature is however
not manifest in curvature scalars. 
In the language of~\cite{ellisking}, it is a
 ``whimper''.

The above model can easily be criticised in that it does not
allow for any sort of backscattering of radiation.
This fact, coupled perhaps with the general expectation that
 ``whimper'' singularities are unstable~\cite{ellisking},
led to the null geometry of the singularity  not receiving the attention it deserved.
In seminal 
work, Poisson and Israel~\cite{pi:is, pi:ih} ``irradiated'' the above spacetime
with a second
``outgoing'' null dust, starting at a retarded time subsequent
to the event horizon. The future evolution can no longer be
represented  as
an explicit solution, but Poisson and Israel were nonetheless
able to deduce that
the Hawking mass would diverge on the null boundary as soon
as one entered the domain of influence of the second null dust--hence
\emph{``mass inflation''}. The divergence of the Hawking mass
implied in particular that the Kretschmann scalar diverged--thus, 
the null boundary turned into a stronger ``curvature singularity'' (in the language
of~\cite{ellisking}).

It turned out, however, that the most interesting aspect of the singularity
of Poisson and Israel~\cite{pi:is,pi:ih} was not that it was ``stronger''
than that of   Hiscock~\cite{hiscock}, but rather,
as shown by Ori~\cite{ori0}, that it too was ``weak'', now from the point
of Tipler's classification~\cite{tipler}. 
Specifically, Ori~\cite{ori0} studied a slightly more tractable
model where the ingoing null dust
was replaced by a shell, and in the context of this model,
showed that the backreaction \emph{would not be sufficient} to immediately
contract the area of the symmetry spheres to $0$, in fact, the metric would remain
well behaved on the null boundary and could continue beyond.
This implies in particular that macroscopic objects might not be ``destroyed''
by the singularity.

In retrospect, as we shall see, the proper characterization of
the singularity in all of the above models~\cite{hiscock,pi:is,ori0}
is not at the level
 of blow up of the 
 curvature but at the level of the non-square integrability of the  Christoffel
 symbols. This notion is perhaps less familiar than the traditional classifications
 of singularity~\cite{ellisking, tipler}, but more relevant
 to pde properties governing the dynamics of the Einstein 
 equations as well as the 
 equations of continuum matter, for which the pointwise blow-up of curvature per se
is of no particular significance~\cite{krs}. 
 From this point of view, all above examples
 are in fact at the \emph{same} level
 of singularity, which,
 though fully deserving of the appelation ``weak'' in view of
 the properties discovered by Ori~\cite{ori0}, is at the same time
 sufficiently strong so as to be ``terminal''. 
 We shall return to this issue after discussing the results for
 Einstein--Maxwell--real scalar field model in Section~\ref{proofofmassinflation}.
See in particular footnote~\ref{uposnmeiwsn}.

\subsubsection{The Einstein--Maxwell--real scalar field model}
\label{themodels}
While better than the Vaidya model of~\cite{hiscock}, the 
above models~\cite{pi:is, ori0} are still open to criticism.
Backscattering is put in by hand, and only occurs strictly after horizon formation;
i.e.~the above spacetimes coincide with the model of~\cite{hiscock}
in  a small neighborhood of $i^+$.
Moreover, the rate of decay of radiation before  this artificial backscattering occurs
is also put in by hand. 

It was clear (see the discussion in~\cite{israel}) that a more realistic model  was necessary. 
In this
regard, a natural candidate was the
\emph{Einstein--Maxwell--real scalar field model}.
Here the Einstein equations are coupled to the Maxwell and
wave equations, but the scalar field is uncharged, thus the latter
two interact only through the gravitational coupling (see Section~\ref{THEMODELSEC}).

The above model admits true wave-like behaviour and allows for the decay
rate of radiation on the horizon to arise dynamically.
It is a generalisation of the spherically symmetric
Einstein--scalar field model whose
mathematical study was initiated by Christodoulou leading
to his seminal proofs of both weak and strong cosmic censorship~\cite{naksinginst}.
In contrast to Christodoulou's model, however, the
Einstein--Maxwell--real scalar field equations 
have a peculiar feature: In view of the absence
of charged matter, then for    the Maxwell tensor to be non trivial,
a complete spacelike hypersurface must have \emph{two} asymptotically
flat ends, just like in maximally extended Schwarzschild or Reissner--Nordstr\"om.
\emph{We will return to this global feature in Section~\ref{globalstarts} 
as this will be central
to considerations of the present paper.}
These global features \emph{per se} do not however directly
relate to the  question of  existence of weak null singularities
arising from the ``point'' $i^+$ of the Penrose diagramme, which
can be addressed restricting
to an arbitrary small neighbourhood of 
 $i^+$.

\subsubsection{The numerical and heuristic studies}
The above model was first approached numerically, starting
with a study by  Gnedin and Gnedin~\cite{gnedingnedin}.
Though in retrospect, the numerics of~\cite{gnedingnedin}  are best viewed
as inconclusive
(see the discussion of~\cite{BDIM, israel}), the authors  themselves interpreted their
results as suggesting that the entire singular boundary is spacelike, casting doubt on 
the mass-inflation/weak null  singularity  scenario.
The heuristic study of the scalar field model by 
Bonnano~\emph{et al}~\cite{BDIM}, on the other hand,
gave  arguments in favour of
the scenario.
This was followed
finally by more careful numerical work
by Brady--Smith~\cite{bradsmith}, starting from data on the horizon, and later 
Burko~\cite{burko} starting
from  characteristic data on a cone terminating on null infinity;
the results of these later numerics supported the  picture
proposed by Poisson--Israel and Ori of a non-empty
piece of null boundary emanating
from $i^+$ on which the mass blows up but across which the metric
is continuously extendible.

\subsubsection{The proof of the mass-inflation/weak null singularities scenario}
\label{proofofmassinflation}
The mass-inflation scenario and the associated weak null singularities were finally
mathematically proven
in~\cite{md:si, md:cbh}  
to indeed occur in the Einstein--Maxwell--real scalar field
model. 

The first result concerns
the persistence of a piece of null boundary of spacetime:
\begin{theorem}[Cauchy horizon stability \cite{md:cbh}]
\label{thes1}
For arbitrary asymptotically flat spherically symmetric
data for the Einstein--Maxwell--real scalar field
system for which the scalar field decays suitably
at spatial infinity $i^0$, then
if the charge is non-vanishing and the event horizon $\mathcal{H}^+$
is asymptotically subextremal,
it follows that the Penrose diagramme contains a subset which is as below
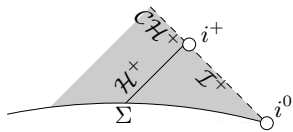
\[
\input{glob_pic.pstex_t}
\]
where $\mathcal{CH}^+$ is a non-empty piece of null boundary. 
Moreover, the spacetime can be continued beyond  $\mathcal{CH}^+$ to a strictly
larger manifold
with $C^0$
Lorentzian metric, to which the scalar field also extends continuously.
\end{theorem}
The $C^0$ extendibility statement captures the ``weakness'' of the singularity 
discussed by Ori in the context of~\cite{ori0}.
The above thus suggests that if  by `sufficiently regular' in the formulation
of strong cosmic censorship, one only requires
that the metric extend $C^0$ (see~\cite{sings}), then the conjecture is false.

The proof of the above mathematical theorem appeals in particular to the fact that 
Price's law   tails hold as an {\bf upper} bound for the behaviour of
the scalar field on the horizon $\mathcal{H}^+$, a statement which in turn was proven in
joint work with Rodnianski~\cite{dr1}. See Theorem~\ref{plt} of Section~\ref{persistence}. 
If such a  polynomial decay rate is also assumed as a {\bf lower} bound
(cf.~the numerical work~\cite{gundlach} and Conjecture~\ref{isitgen} in Section~\ref{globsingsec} below), then one obtains that 
the null boundary piece $\mathcal{CH}^+$ is in fact singular:
\begin{theorem}[Mass inflation \cite{md:cbh}]
\label{thes2}
If a suitable {\bf lower} bound on the decay rate of the scalar field on the event horizon 
$\mathcal{H}^+$ is assumed,
then the non-empty piece of null boundary $\mathcal{CH}^+$ of
Theorem~\ref{thes1} is in fact a weak null singularity on which the Hawking mass
blows up identically. In particular, the metric cannot
be continued beyond $\mathcal{CH}^+$  as a $C^2$ metric. The scalar
field cannot be extended beyond $\mathcal{CH}^+$ as a $H^1_{loc}$ function.
\end{theorem}

We have stated above the $C^2$ inextendibility result for the metric because it is the
easiest to give a self-contained
proof in view of the fact that curvature is a geometric invariant.
One expects to be able to make a  much stronger statement, however, namely that  \emph{the
metric can in fact not be continued beyond $\mathcal{CH}^+$  as a continuous
metric with
Christoffel symbols  in $L^2_{\rm loc}$.}\footnote{In~\cite{md:cbh}, it is shown that
the Christoffel symbols fail to be in $L^2_{\rm loc}$ for a \emph{particular}  extension 
with continuous metric.} This notion of inextendibility,
though not sufficient to show that macroscopic observers are torn apart,
ensures that the boundary of spacetime is singular enough so that
one cannot continue as
weak solutions to a suitable Einstein--matter system. In this sense,
it is sufficient to ensure a version of determinism.
The above results suggest that ``inextendiblity as a Lorentzian manifold
with continuous metric and with Christoffel symbols in $L^2_{\rm loc}$'' may be the correct  formulation
of  ``inextendible as a suitably regular Lorentzian metric'' in the statement of
 strong cosmic censorship. This formulation is due
 to Christodoulou.  See the discussion in the introduction
of~\cite{formation}.\footnote{\label{uposnmeiwsn}We note that these 
inextendibility remarks apply equally to the
models of Section~\ref{nulldustmodels}. In particular, from this
point of view, the singularity in the model studied by
 Hiscock~\cite{hiscock} is not really less singular than that of 
Poisson--Israel~\cite{pi:is}, when both are examined primarily from the perspective 
of the blow-up of their Christoffel symbols.
 In this regard, see also the discussion in~\cite{hermanhiscock}.}

We remark finally that
there is an additional level in the hierarchy of spherically symmetric
models which one can consider--namely
the Einstein--Maxwell--\emph{complex} scalar field system.
In this model, the scalar field carries charge which can source
the Maxwell field, and thus charged black holes
can form in collapse
from data which are topologically Euclidean and contain
no horizons.
This model has been studied numerically in~\cite{hodpiran} while
results similar to Theorem~\ref{thes1} (conditional, however, on an analogue
of Theorem~\ref{plt}) have been obtained
in upcoming work of Kommemi~\cite{kommemiCauchy}.

\subsection{Global, bifurcate weak null singularities}
\label{globalstarts}
As we have already discussed in Section~\ref{themodels},
in contrast to the model of~\cite{hodpiran, kommemi, kommemiCauchy}
just referred to, 
the Einstein--Maxwell--real scalar field system 
is such that for the Maxwell tensor to be non-trivial,
complete initial data necessarily will have
two asymptotically flat ends.  The theorems of the previous section 
only probed the structure of the boundary of spacetime
in a neighbourhood of $i^+$. \emph{What about the remaining boundary?}
This will be the subject of the present paper.

A preliminary result, 
using the fact that  the  matter model is, in the language of~\cite{kommemi},
``strongly tame'',
implies that, if the initial data hypersurface $\Sigma$ is moreover assumed
to be ``future admissible'' (see Section~\ref{genframsec}),
this boundary in general 
is as below:
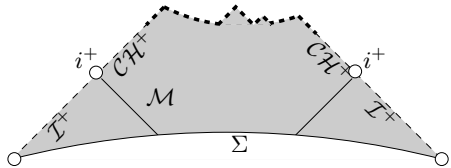
\[
\input{glob_gen.pstex_t}
\]
where in addition to the null boundary components $\mathcal{CH}^+$ emanating
from $i^+$, on which $r$ is bounded below (at this level of generality,
\emph{these components are 
possibly empty}, but are indeed non-empty if Theorem~\ref{thes1}
applies),
there is an (again, \emph{possibly empty!})~achronal boundary on which
$r$ extends continuously to $0$, depicted above as the thicker-shaded dotted line. See 
Proposition~\ref{frameprop} for a precise statement.
 Note that the boundary decomposition implicit in the above diagramme already 
 represents a non-trivial statement about
possible singularity formation, 
as it excludes in particular ``first'' singularities corresponding to TIPs with 
compact intersection with $\Sigma$ and whose spherically
symmetric spheres do not
contract in the limit to zero area.
In view of the topology of the initial data and the monotonicity properties
inherent in the Einstein equations, it follows that
$J^-(\mathcal{I}^+)$ has (as depicted) a non-empty complement and $\mathcal{I}^+$
is complete. Thus, in particular, the above result contains the
statement of  ``weak'' cosmic censorship.\footnote{We emphasise, however,
that this is a \emph{very} soft result in comparison with Christodoulou's seminal 
proof~\cite{naksinginst} of
weak cosmic censorship in the case of one end, where the main
obstacle is null boundary components emanating from the centre of symmetry. In
the two-end case, there is no such centre and thus, a fortiori, no null boundary
components of the above type.}

In this short paper, it is shown that when one restricts
to suitably small perturbations of sub-extremal
Reissner--Nordstr\"om, the $r=0$ boundary 
{\bf is in fact empty}, and the two null components join 
to form a bifurcate null boundary across which  the spacetime is globally
extendible as a manifold with $C^0$ Lorentzian metric (to which the scalar
field also extends continuously)! Moreover, 
if the assumption of Theorem~\ref{thes2} holds on both components
of the event horizon $\mathcal{H}^+$, then the Hawking mass blows up
\emph{identically} on this bifurcate null boundary and the spacetime is 
\emph{inextendible}
as a Lorentzian manifold with  $C^2$ metric, the
scalar field is  intextendible as an $H^1_{\rm loc}$ function, and
in fact, one can make a statement about the blow up 
of the $L^2$ norm of the
Christoffel symbols.  
Specifically:
\begin{theorem}
\label{introtheo}
Let $(\mathcal{M},g, \phi, F)$ be the maximal development
of  sufficiently small 
spherically symmetric perturbations of asymptotically flat two-ended data corresponding
to subextremal
Reissner--Nordstr\"om with parameters $0<Q_{RN}<M_{RN}$, under the evolution of
the Einstein--Maxwell--real scalar field system.

Then there exists a later Cauchy surface $\Sigma_+$ which is
future-admissible and such that to the future
of $\Sigma_+$,
 the Penrose diagramme of $(\mathcal{M},g)$ is given by:
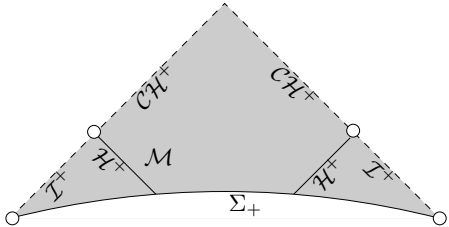
\[
\input{glob_pert.pstex_t}
\]

Similarly, there exists also an earlier Cauchy surface $\Sigma_-$
which is past-admissible,
such that to the past of $\Sigma_-$, the Penrose diagramme of $(\mathcal{M},g)$
is given by a time reversed depiction of the above,
with boundary components $\mathcal{I}^-$,
$\mathcal{CH}^-$, etc.

The global bound
 \[
 r\ge M_{RN}-\sqrt{M^2_{RN}-Q^2_{RN}}-\epsilon
 \]
 holds 
 for the area-radius $r$ of the spherically symmetric spheres,
where $\epsilon\to 0$
as the `size' of the perturbation tends to $0$.
Moreover, the metric extends continuously beyond
 $\mathcal{CH}^+$ to a strictly larger Lorentzian manifold 
 $(\widetilde{\mathcal{M}},\widetilde{g})$, making $\mathcal{CH}^+$ a bifurcate null 
hypersurface in $\widetilde{\mathcal{M}}$. The scalar field $\phi$
extends to a continuous function on $\widetilde{\mathcal{M}}$.
All future-incomplete  
causal geodesics in $\mathcal{M}$ extend to enter $\widetilde{\mathcal{M}}$.

Finally, if $\phi$ satisfies the assumption of Theorem~\ref{thes2}  
 on both components of the horizon $\mathcal{H}^+$, then the Hawking mass 
extends ``continuously''  to $\infty$ on all of $\mathcal{CH}^+$.
In particular,   $(\mathcal{M},g)$ is future inextendible as 
a  spacetime with $C^2$ Lorentzian metric,
and continuous extensions $\phi$ as in the previous paragraph fail to be $H^1_{\rm loc}$  
in a neighborhood of every point of  $\partial\mathcal{M}\subset\widetilde{\mathcal{M}}$.

Exactly analogous statements hold for $\mathcal{H}^-$, $\mathcal{CH}^-$.
\end{theorem}

In fact, one can hope to show that, under the assumption of Theorem~\ref{thes2},
given an extension $(\widetilde{\mathcal{M}},\widetilde{g})$ with 
$C^0$ metric, 
the Christoffel symbols fail to be $C^{0,1}$ 
everywhere on $\mathcal{CH}^+$, and,
wherever $\phi|_{\mathcal{CH}^+}$ is not locally constant,
the Christoffel symbols in fact fail to be  locally $L^2$.
In particular, the latter holds in a neighbourhood of $i^+_{A,B}$.
As in~\cite{md:cbh}, we will here, however, only show these statements \emph{for a particular $C^0$ extension} and
not address the geometric issues in generalising this for all such extensions.

If one only considered, say, future evolution,
the proof of the stability part of 
the above theorem would be much easier (but much less interesting)
if one moreover restricted to data which is supported strictly between
the two components of the future event horizon $\mathcal{H}^+$, as such solutions would
coincide with Reissner--Nordstr\"om in a neighbourhood of $i^+$. Evolving to the past, 
however,
this property would be lost. We have stated our result as a statement
about both future and past evolution to emphasise that we
are not dealing with that essentially uninteresting case.
 In fact, our smallness condition does not require that
$\phi$ be initially of compact support.

The analogy of the Reissner--Nordstr\"om and Kerr families suggests that the situation for 
solutions of the vacuum equations (with no symmetry
assumed!) in a neighbourhood of
the Kerr family may be similar.  This is further supported by a perturbative
analysis by Ori~\cite{ori}.
One may thus conjecture:
\begin{conjecture}[A.~Ori]
\label{introconj}
Let $(\mathcal{M},g)$ be the maximal vacuum Cauchy development of 
sufficiently small perturbations of asymptotically flat two-ended 
Kerr data corresponding to parameters $0<|a|<M$.
Then there exist  both a future and past extension $(\widetilde{\mathcal{M}},\widetilde{g})$ 
of $\mathcal{M}$ with $C^0$ metric $\widetilde{g}$
such that $\partial\mathcal{M}$ is a bifurcate
null cone in $\widetilde{\mathcal{M}}$ and {\bf all} future (past)  incomplete geodesics in
$\gamma$ pass into $\widetilde{\mathcal{M}}\setminus\mathcal{M}$.

Moreover, for generic such perturbations, any $C^0$ extension $\widetilde{\mathcal{M}}$
will fail to have $L^2$ Christoffel symbols in a neighbourhood of any
point of $\partial\mathcal{M}$.
\end{conjecture}

Thus, according to the above conjecture, not only
should black holes generically have part\footnote{We emphasise explicitly that the implication that \emph{part} of the
singular boundary be null, which is already suggested by Theorem~\ref{thes1}
and~\ref{thes2}, should apply to generic black holes
forming in collapse from one asymptotically flat end, in view of the
fact that it is believed that these always approach Kerr in a neighbourhood
of $i^+$ in the domain of outer communications.
The above conjecture
cannot however address the {\bf global} structure of the 
boundary of such spacetimes as these are 
{\bf globally} far from the Kerr family. We note, for instance, that
the numerical work of~\cite{hodpiran} in the context of the spherically symmetric
charged scalar field model suggests, in addition to a null component,
also a non-empty spacelike portion of 
the singularity in the case of collapse of data with one end. 
See also Conjecture~1.11 of~\cite{kommemi} and subsequent comments therein.} 
of their singular boundary 
null, but
a small neighbourhood  (in the moduli space of vacuum metrics)
of the Kerr family would, in both future and past,
  generically {\bf be entirely free of spacelike singularities}.

We end this introduction with a final remark: Until a few years ago,
``large data'' problems for the Einstein vacuum equations without symmetry
seemed mathematically   intractable.
This view changed completely with the seminal work of Christodoulou~\cite{formation} on
black hole formation in  vacuum,
which demonstrated how largeness can be ``tolerated'' in the analysis as long as    it 
appears
in a controlled way and is ``coupled''  with a corresponding
``smallness''. In broad terms, this is in fact quite reminiscent of structure that plays a role
in  the analysis here.
An even more directly related  link is provided by the recent
Luk--Rodnianski theory~\cite{LukRod1, LukRod2} of interacting impulsive gravitational waves, reviewed briefly in  Section~\ref{lukrodtheory} of this paper,
which illuminates a new hiearchy in the vacuum
equations which is almost certainly  key to
 understanding solutions in the presence of null singularities.
 These  works, coupled with progress on the stability of black hole
 \emph{exterior} regions (surveyed in~\cite{bhsp}, cf.~the role of Theorem~\ref{plt}), give hope that Conjecture~\ref{introconj}, 
 and its      ``cosmological twin'' Conjecture~\ref{cosmocon}--but also, quite ominously,
 Conjecture~\ref{cosmocon2}--both  
 stated in Section~\ref{epilogue},
will soon  be amenable to mathematical analysis.

\section{The Einstein--Maxwell--(real)-scalar field model}
\label{THEMODELSEC}
The Einstein--Maxwell--(real)-scalar field system describes the
interaction of a gravitational field, an electromagnetic field, and an uncharged scalar field.
The latter two thus interact solely via their coupling to gravity.
The equations are given by
\begin{equation}
\label{EMSF1}
R_{\mu\nu}-\frac12g_{\mu\nu}R=8\pi T_{\mu\nu},
\end{equation}
\begin{equation}
\label{EMSF2}
\nabla^\mu F_{\mu\nu} =0 , \qquad \nabla_{[\lambda}F_{\mu\nu]}=0,
\end{equation}
\begin{equation}
\label{EMSF3}
\Box_g\phi = 0,
\end{equation}
\begin{equation}
\label{EMSF4}
T_{\mu\nu}= \partial_\mu\phi\,\partial_\nu\phi-\frac12 g_{\mu\nu} \nabla^\alpha\phi\nabla_\alpha\phi
+\frac{1}{4\pi}F^\alpha_{\, \, \mu}F_{\alpha\nu}- \frac1{16\pi} g_{\mu\nu} F^{\alpha\beta} F_{\alpha\beta}.
\end{equation}

Let us recall the version of the foundational
well posedness theorem for general relativity, which
applies for the above model: 
\begin{litthe}[\cite{Mch}]
Let $(\Sigma, \bar{g}, K, \bar{F}_{ab}, \bar{F}_{0b}, \phi_0,\phi_1)$   be a smooth
Einstein--Maxwell--scalar field
 data set. Then there exists a unique
non-empty maximal smooth 
Cauchy development  $(\mathcal{M},g,F_{\mu\nu}, \phi)$
for the equations $(\ref{EMSF1})$--$(\ref{EMSF4})$.    
\end{litthe}

The motivation for the above system has been discussed already
in Section~\ref{themodels} and at length  in~\cite{md:si, md:cbh}.
It is perhaps the simplest system admitting 
Reissner--Nordstr\"om as an explicit solution and allowing at the same time
for non-trivial wave dynamics in spherical symmetry.

\section{The spherically symmetric reduction}
\label{ssred}
In the case of spherically symmetric data, then the maximal Cauchy development
is spherically symmetric and the Einstein equations reduce to a $1+1$
dimensional system.

For the purpose of this discussion,
by spherically symmetric data set
 we assume that $(\Sigma, \bar{g})$
is a warped product $\mathbb R\times_r \mathbb S^2$,
with metric $dx^2 +r^2(x)d\sigma_{\mathbb S}^2$ (admitting the obvious ${\rm SO}(3)$-action),
$K_{ab}$ is an ${\rm SO}(3)$-invariant symmetric $2$-tensor on 
$\mathbb R\times_r \mathbb S^2$,
$\phi_0$, and $\phi_1$ are functions on $\mathbb R$,
and $\bar{F}_{ab}$, $\bar{F}_{0b}$ 
are an ${\rm SO}(3)$ invariant $2$-form and vector field, respectively, on
$\mathbb R\times_r \mathbb S^2$, together satisfying the
Einstein--Maxwell-real scalar field \emph{constraint}
equations.  These equations can be found, for instance, in~\cite{dr1}.

\begin{proposition}
Let
$(\mathcal{M},g, F, \phi)$ be the maximal Cauchy development of
smooth spherically
symmetric data as described above. Then $\mathcal{M}$ is spherically symmetric
and
is moreover
of the form of a warped product $(\mathcal{M},g) =\mathcal{Q}\times_{r}\mathbb S^2$, 
where $\mathcal{Q}$
admits a global null coordinate system $(u,v)$,
and the metric
takes the form
\[
-\Omega^2(u,v)\, du\, dv+r^2(u,v)\sigma_{\mathbb S^2}.
\]
The Maxwell field takes the form    
\[
F= \frac{\Omega^2 Q_e}{2r^2} du\wedge dv+  Q_m \sin\theta\, d\theta\wedge d\varphi
\]
for real constants $Q_e$, $Q_m$,
and $\phi$ descends to a function $\phi:\mathcal{Q}\to \mathbb R$. 
Defining $Q=\sqrt{Q_e^2+Q_m^2}$, the
full content of the equations is thus given by
\begin{equation}
\label{requ}
\partial_u\partial_v r=-\frac{\Omega^2}{4r}-\frac{1}r\partial_v r\partial_u r+\frac14
\Omega^2r^{-3}Q^2,
\end{equation}
\begin{equation}
\label{Omegeq}
\partial_u\partial_v \log \Omega^2=-4\pi \partial_u\phi\partial_v\phi+\frac{\Omega^2}{4r^2}+\frac1{r^2}
\partial_vr\partial_ur-\frac{\Omega^2Q^2}{2r^{4}},
\end{equation}
\begin{equation}
\label{WAVe}
\partial_u (r\partial_v \phi)=-\partial_u\phi\partial_vr,
\end{equation}
\begin{equation}
\label{Raych1}
\partial_u(\Omega^{-2}\partial_u r) =-4\pi r \Omega^{-2}(\partial_u\phi)^2,
\end{equation}
\begin{equation}
\label{Raych2}
\partial_v(\Omega^{-2}\partial_v r) = -4\pi r \Omega^{-2}(\partial_v\phi)^2.  
\end{equation}
\end{proposition}

 Let us note that we have a gauge freedom in defining null coordinates
 $\bar{u}= f(u)$, $\bar{v}= g(v)$, for general smooth functions $f$, $g$. 
 We will normalise expedient null coordinate systems later
 on.
 
 In ``translating'' intuition from the above system to the problem 
 of vacuum collapse, 
the reader familiar with the formalism of~\cite{formation}
should note that $r\partial_v r$ corresponds to ${\rm tr} \chi$,
while it is useful to think of $r\partial_v \phi$ as  an analogue of
$\hat\chi$, even though in the spherically symmetric model, $r\partial_v\phi$
appears in the geometry (via the coupling) at the level of curvature.
 
\section{The general framework}
\label{genframsec}
We will consider spherically symmetric, asymptotically flat data
$(\Sigma, \ldots)$ with two  ends.  Let us label the ends $A$ and $B$,
and let us restrict our $(u,v)$ coordinate system such that 
asymptotically on $\Sigma$,
$\partial_v$ points towards end $A$, and $\partial_u$ points towards $B$.

In order to obtain a ``good''  a priori
characterization
of the boundary of the maximal development, it is convenient to restrict to data
such that 
\begin{equation}
\label{admissibility}
\Sigma= {\Sigma_A}\cup\Sigma_ B
\end{equation}
where $\Sigma_A$ satisfies $\partial_ur<0$, $\Sigma_B$ satisfies
$\partial_vr<0$, and $\Sigma_A$, $\Sigma_B$ are connected. 
We shall call  two-ended asymptotically flat data satisfying $(\ref{admissibility})$
\emph{future admissible}. 
Let us note that asymptotic flatness requires that the end $A$ is ``contained''
in $\Sigma_A$ and disjoint from $\Sigma_B$, and similarly, the end
$B$ is ``contained'' in $\Sigma_B$ and disjoint from $\Sigma_A$.

Note that Cauchy hypersurfaces in Schwarzschild or Reissner--Nordstr\"om
are future admissible if they have no intersection with the closure of the white hole region in 
$\mathcal{Q}$.

The significance of the decomposition $(\ref{admissibility})$
is that by Raychaudhuri $(\ref{Raych1})$, $(\ref{Raych2})$ it follows
that $\mathcal{Q}\cap J^+(\Sigma)=\{\partial_ur<0\}\cup \{\partial_vr<0\}$.
It serves as a generalisation of Christodoulou's ``no antitrapped surfaces'' 
condition~\cite{sings},
appropriately reworked for the two-ended case.
Note finally that the property $(\ref{admissibility})$ is stable
to perturbation of data.

For future admissible data we have the following result which in fact follows from
estimates shown in~\cite{md:cbh}, but is best
 viewed as a special
case of more general results of Kommemi~\cite{kommemi} 
concerning the Einstein--Maxwell--complex scalar field system:

\begin{proposition}
\label{frameprop}
Let $(\mathcal{M}=\mathcal{Q}\times_r\mathbb S^2, \Omega^2, \phi, Q_e, Q_m)$ 
be the maximal Cauchy development of data as above, with $\Sigma$ future
admissible. 

I.~Then, in the notation of~\cite{kommemi},
the Penrose diagramme of $(\mathcal{M},g)$ 
to the future of $\Sigma$ is as below:
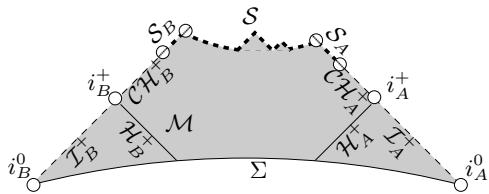
\[
\input{glob_form.pstex_t}
\]
with the usual convention that some of the boundary components may be empty.

II.~Under the above assumptions, then
$\mathcal{I}^+_{A,B}$ is complete,  the horizons $\mathcal{H}^+_{A,B}$ 
are non-empty  and a Penrose inequality 
$\sup_{\mathcal{H}^+_{A,B}}r\le 2M^f_{A,B}$ holds on each horizon, where
$M^f_{A,B}$ is the infimum of the Bondi mass of $\mathcal{I}^+_{A,B}$,
respectively.

III.~Suppose $(\widetilde{\mathcal{M}},\widetilde{g})$ is a $C^2$ future-extension of
$(\mathcal{M},g)$. Then, denoting
by $\pi_Q:\mathcal{M}\to\mathcal{Q}$ the natural projection,
there exists a  causal curve $\gamma$
passing into $\widetilde{\mathcal{M}}\setminus\mathcal{M}$ such 
that $\overline{\pi_{\mathcal{Q}}(\gamma|_{\mathcal{M}})}\cap (\mathcal{CH}^+_A\cup\mathcal{CH}^+_B)\ne\emptyset$.

\end{proposition}

We review very briefly some properties encoded in the above
notation (see~\cite{kommemi} for details). 
The boundary segments $\mathcal{I}^+_{A,B}$ have the
property that $r$ extends ``continuously'' to $\infty$. We identify
$\mathcal{I}^+_{A,B}$ with ``future null infinity''. The notation
$i^+_{A,B}$ in fact already encodes the statement that ``future null infinity is
complete'' (cf.~Part II above). Monotonicity properties (cf.~Section~\ref{trappedsec} 
below) allow for the definition of a nonnegative Bondi mass function on
$\mathcal{I}^+_{A,B}$ whose infimum defines the final Bondi mass $M^f_{A,B}$. 
The sets $\mathcal{CH}^+_{A,B}$ are half-open segments emanating from
(but not including) ${i}^+_{A,B}$ characterized
by the property that $0$ is not a limit point of $r$ on their interior,
whereas $r$ extends continuously to $0$ on
$\mathcal{S}_A\cup\mathcal{S}_B\cup
\mathcal{S}$.

Note that since by Part II, $J^-(\mathcal{I}^+)$ has indeed a non-trivial
complement, we have $\mathcal{CH}^+_A\cup\mathcal{CH}^+_B\cup\mathcal{S}_A
\cup\mathcal{S}_B\cup
\mathcal{S}\ne\emptyset$. Any $4$ of the above $5$ components, however,
can be empty.

We emphasise that Proposition~\ref{frameprop} is in fact a very general result
for spherically symmetric Einstein--matter systems,  holding for all 
matter-models which, in the terminology introduced by Kommemi~\cite{kommemi},
are  ``strongly tame''.
Another example of a strongly tame matter model is the Einstein--Vlasov
system. See~\cite{dafren}.

Note finally that a  Cauchy hypersurface which is future
admissible as we have defined it
cannot be admissible to past evolution.  In Reissner--Nordstr\"om or Schwarzschild,
then Cauchy hypersurfaces are past-admissible iff they do not intersect the closure
of the black hole region.
On the other hand, by a Cauchy stability
argument, one can show that given an \emph{arbitrary} spherically symmetric
asymptotically flat
Cauchy hypersurface of Reissner--Nordstr\"om, then for suitably
small perturbations of initial data, there will exist a future admissible 
Cauchy surface in the future, and a past-admissible Cauchy surface in the past
which coincide for large $r$ with the original hypersurface.
See Corollary~\ref{admississue} of Section~\ref{CSsec}.

\section{The trapped region, event horizons, and outermost apparent horizons}
\label{trappedsec}
We assume in this section that $\mathcal{Q}$ is given by Proposition~\ref{frameprop}.
We will restrict consideration to $J^+(\Sigma)$, i.e.~by $\mathcal{Q}$ below
we shall actually mean $J^+(\Sigma)\cap \mathcal{Q}$.
The notation $J^-(\mathcal{I}^+)$, etc., is understood in the usual sense
i.e.~as $J^-(\mathcal{I}^+)\cap \mathcal{Q}(\cap J^+(\Sigma))$, 
where $J^-$ denotes now the
causal structure of the ambient $1+1$ dimensional Minkowski space defining
the Penrose diagrammes.
Let us note that for $p\in \mathcal{Q}$, then $J^+_{\mathcal{Q}}(p)=
J^+_{\mathbb R^{1+1}}(p)\cap \mathcal{Q}$ so there is no
danger of confusion. Closure on the other hand will be always with respect to the ambient
topology of $\mathbb R^{1+1}$.

Let us define the \emph{trapped region}:
\[
\mathcal{T}=\{(u,v)\in\mathcal{Q}:\partial_ur<0, \partial_vr<0\}.
\]
Our definition of admissibility
$(\ref{admissibility})$
 implies that $\mathcal{T}\cap \Sigma\neq\emptyset$,
in particular,
\[
\mathcal{T}\neq \emptyset.
\]

The Raychaudhuri equations $(\ref{Raych1})$, $(\ref{Raych2})$ then
imply that $\partial_v r >0$ in $J^-(\mathcal{I}^+_A)$, 
$\partial_u r>0$ in $J^-(\mathcal{I}^+_B)$.
We recall the definition of $\mathcal{H}^+_A$ as the future boundary
of $J^-(\mathcal{I}^+_A)$ in $\mathcal{Q}$,
and similarly $\mathcal{H}^+_B$. It follows that $\partial_vr \ge 0$
on $\mathcal{H}^+_A$, while $\partial_ur\ge 0$ on $\mathcal{H}^+_B$.

By our definition of admissibility, we easily see that 
 on $\overline{J^-(\mathcal{I}^+_A)}\cap\mathcal{Q}$, 
 and thus, by Raychaudhuri $(\ref{Raych1})$ 
 also on $J^+(\mathcal{H}^+_A)$,
 we must have $\partial_ur<0$, while
 similarly, on $\overline{J^-(\mathcal{I}^+_B)}\cap\mathcal{Q}$, 
 and thus, by $(\ref{Raych2})$,
 also on $J^+(\mathcal{H}^+_B)$, we must have $\partial_vr<0$.

Let us finally note 
that in view of the above,
we  indeed have
$\mathcal{H}^+_A\cap\mathcal{H}^+_B=\emptyset$,
as has been depicted and
\[
\mathcal{T}\cap \overline{J^-(\mathcal{I}_{A,B}^+)}=\emptyset.
\]

We define $\mathcal{A}'_{A}\subset  J^+(\mathcal{H}^+_A)$ to be the set
$(u,v)\in  J^+(\mathcal{H}^+_A)$ such that $\partial_v r(u,v)=0$ but $\partial_v r(u^*,v)>0$ for
$u^*<u$. We call $\mathcal{A}'_A$ the \emph{outermost
apparent horizon} corresponding to end $A$.

Our assumptions in fact imply the non-emptiness of $\mathcal{A}'_A$.
For this, we first notice that, in view of both $(\ref{Raych1})$, $(\ref{Raych2})$,
we have
$\mathcal{CH}^+_B\cap \overline{J^+(\mathcal{H}^+_A)}\subset\overline{\mathcal{T}}$.
Similarly, we have
\[
\mathcal{S}_A\cup\mathcal{S}_B\cup \mathcal{S}\subset \overline{\mathcal{T}}.
\]
(Indeed, if  $r(u,v)< \inf_{x\in \Sigma} r(x)$, then $(u,v)\in \mathcal{T}$.)
It follows that along every ingoing null ray from $\mathcal{H}^+_A$,
one must encounter a unique point of $\mathcal{A}'_A$.

We define $\mathcal{A}'_B$ analogously.

We emphasise that $\mathcal{A}'_{A,B}$ as defined can be seen to be achronal sets,
but are not necessarily connected.

Finally, let us remark that the above notion of trapped region $\mathcal{T}$ 
is defined by restricting attention to surfaces of symmetry. 
For a  general discussion of \emph{non-spherically symmetric} trapped surfaces
in spherically symmetric ambient spacetimes, see~\cite{seno}.

\section{The Hawking mass}
Recall the important quantity
\begin{equation}
\label{mdef}
m=\frac r2(1-|\nabla r|^2)=\frac r2(1+4\Omega^{-2}\partial_ur\partial_vr)
\end{equation}
known as the \emph{Hawking mass}.

We remark that the Hawking mass is at the level of first derivatives  of
the metric.
We also note, however, the inequality
\begin{equation}
\label{Kretsch}
R_{\mu\nu\alpha\beta} R^{\mu\nu\alpha\beta} \ge \frac{4}{r^4}
\left(\frac{2m}r\right)^2
\end{equation}
for the Kretschmann scalar.

The so-called \emph{renormalised Hawking mass} 
\[
\varpi = m+\frac{Q^2}{2r}
\]
satisfies
\begin{equation}
\label{RHM1}
\partial_u \varpi = - 8\pi r^2\Omega^{-2}\partial_v r (\partial_u\phi)^2, 
\end{equation}
\begin{equation}
\label{RHM2}
\partial_v \varpi =-8\pi r^2\Omega^{-2}\partial_u r (\partial_v\phi)^2.
\end{equation}

We see the monotonicity properties 
\begin{equation}
\label{monintrap}
\partial_u\varpi\ge 0,\qquad \partial_v\varpi\ge 0\qquad{\rm\ in\ }\mathcal{T},
\end{equation}
while
\[
\partial_u\varpi\le 0, \qquad \partial_v\varpi\ge 0\qquad{\rm\ in\ }\overline{J^-(\mathcal{I}^+_A)}\cup
(J^+(\mathcal{H}^+_A)\cap J^-(\mathcal{A}'_A)),
\]
and similarly, 
\[
\partial_v\varpi\le 0, \qquad \partial_u\varpi\ge 0\qquad{\rm\ in\ }\overline{J^-(\mathcal{I}^+_B)}\cup
(J^+(\mathcal{H}^+_B)\cap J^-(\mathcal{A}'_B)).
\]
Since by $(\ref{mdef})$ we have $0<r=2m$ on $\mathcal{A}'_{A,B}$,
it follows that on $\mathcal{A}'_{A,B}$ we have
$\varpi\ge Q$ and $\varpi  > 0$.

Defining now the asymptotic renormalised Hawking mass
$\varpi^+_{A,B}$ on each
of the two horizons $\mathcal{H}^+_{A,B}$
by
\[
\varpi^+_{A,B} =\sup_{\mathcal{H}^+_{A,B}} \varpi,
\]
we have 
by our monotonicity properties that
\[
0<\varpi^+_{A,B}\le M^f_{A,B}\le  M^{ADM}_{A,B}<\infty.
\]
where $M^f_{A,B}$, $M^{ADM}_{A,B}$ denote the final Bondi and ADM mass
associated to each end, as well as
\begin{equation}
\label{notsuper}
Q\le \varpi^+_{A,B}.
\end{equation}

Interestingly, $(\ref{notsuper})$ can be proven given only 
the solution in $\mathcal{H}^+\cup J^-(\mathcal{I}^+)$ under the assumption 
that $r$ is bounded above on $\mathcal{H}^+$, i.e.~without
using $\mathcal{A}'_{A,B}\ne\emptyset$.  See~\cite{md:cbh}.

Let us note that $(\ref{notsuper})$ allows for the case
\begin{equation}
\label{extrem}
Q=\varpi^+_{A,B}.
\end{equation}
This is the case of an \emph{asymptotically extremal black hole} (from the perspective
of the event horizon\footnote{cf.~the condition $Q=M^f_{A,B}$ discussed in~\cite{kommemicharge}}). 
 As we shall see, 
Theorem~\ref{chstab} below will have to exclude 
$(\ref{extrem})$ in its statement. It is an open question to understand 
the dynamics sufficiently so as to
characterize the causal structure of the boundary in the case $(\ref{extrem})$, even in a neighbourhood
of $i^+_{A,B}$. See the remarks in~\cite{kommemi}.

A much  more general discussion, including some of the statements given 
here, is contained in Kommemi~\cite{kommemicharge}.

\section{Persistence of the null boundary: $\mathcal{CH}^+_{A,B}\ne \emptyset$}
\label{persistence}
We now give a more precise version of the first theorem
of~\cite{md:cbh} (Theorem~\ref{thes1} of the introduction) 
adapted to the current
framework.

For this, let us first normalise the null coordinates $u$ and
$v$.
Given $\Sigma$ with metric $dx^2+r^2(x)d\sigma_{\mathbb S^2}$,
where $x\to\infty$ corresponds to end $A$, $x\to-\infty$ corresponds to end $B$,
we pick an arbitrary origin and normalise our null coordinates such that
$\partial_ux=-1$, $\partial_vx=1$ on $\Sigma$.
We note that our notion of asymptotic flatness then
guarantees that these coordinates are Eddington--Finkelstein-like, in
the sense that $\partial_ur$ is uniformly bounded
above and away from $0$ on $\mathcal{I}^-_B$, $\mathcal{I}^+_A$,
and similarly $\partial_vr$ on $\mathcal{I}^-_A$, $\mathcal{I}^+_B$ 
(in small neighborhoods of $i^0_{A,B}$).

\begin{theorem}[\cite{md:cbh}]
\label{chstab}
Under the assumptions of Proposition~\ref{frameprop}
and the above normalisation of the null coordinates, 
let us in addition require that initially on $\Sigma$,
\begin{equation}
\label{stoapeiro}
|\partial_u (r\phi)|\le Cr^{-2}, \qquad |\partial_v(r\phi)|\le Cr^{-2},
\end{equation}
for some arbitrary constant $C$.
Then if
\begin{equation}
\label{nonex}
0< Q< \varpi^+_{A,B},
\end{equation}
then 
\begin{equation}
\label{thenonemp}
\mathcal{CH}^+_{A,B}\ne \emptyset.
\end{equation}
Moreover,
there exist  neighourhoods $\mathcal{U}_{A,B}$ of $i^+_{A,B}$ such 
that $\mathcal{A}'_{A,B}\cap \mathcal{U}_{A,B}$ is a connected achronal
curve terminating at $i^+_{A,B}$, and such that $D^+(\mathcal{A}'_{A,B})\subset\mathcal{T}$. 
\end{theorem}

The second statement above signifies that $\mathcal{A}'_{A,B}\cap \mathcal{U}_{A,B}$ is
what is sometimes known as
a `dynamical horizon', except that it may contain outgoing null pieces.
Such null pieces
will in fact be excluded under the assumptions of Theorem~\ref{blowupth}
in the next section.
See~\cite{williams} for various general results about $\mathcal{A}'$
in spherically symmetric spacetimes.

Let us note that 
the proof of Theorem~\ref{chstab} required to first establish
a version of ``Price's law''-type decay for the scalar field along the event horizons
$\mathcal{H}^+_{A,B}$. This was
proven in joint work with Rodnianski~\cite{dr1}. 
As this has independent interest, we give the statement 
relevant for data satisfying $(\ref{stoapeiro})$:
\begin{theorem}[M.D.--Rodnianski \cite{dr1}]
\label{plt}
Under the assumptions $(\ref{stoapeiro})$, $(\ref{nonex})$ of Theorem~\ref{chstab},
then for all $\epsilon>0$,
\begin{equation}
\label{forprice'slaw}
|\phi|+ |\partial_v\phi| \le  \tilde{C}_\epsilon 
v^{-2+\epsilon}, \qquad |\phi|+|\partial_u\phi| \le \tilde{C}_\epsilon u^{-2+\epsilon}
\end{equation}
on $\mathcal{H}^+_A$, $\mathcal{H}^+_B$ respectively,
for a $\tilde{C}_\epsilon$ depending only on $\epsilon$,
 $C$, $\varpi^+_{A,B}$ and $Q$. 
\end{theorem}
For data where $\phi$ is 
of compact support on $\Sigma$, the decay bounds in the statement
of Theorem~\ref{plt} on the right hand
side of $(\ref{forprice'slaw})$ can be
improved to $\tilde{C}_\epsilon v^{-3+\epsilon}$, $\tilde{C}_\epsilon u^{-3+\epsilon}$,
illuminating the familiar obstruction at power $3$. 
It is well known, however, that 
general semilinear problems even on a fixed Minkowski background already lead to
decay which is captured by inverse polynomial decay weaker than $t^{-3}$,
even for compactly supported data.
Thus, it is far preferable to state
Theorem~\ref{plt} 
in the above  more general context of data satisfying only $(\ref{stoapeiro})$,
where the statements are more robust.

Finally, let us remark that 
assumption $(\ref{nonex})$ is a ``teleological'' assumption.
Note, however, that in view of the monotonicity properties of Section~\ref{trappedsec},
it follows immediately that $(\ref{nonex})$ is indeed satisfied if it is assumed
for example that initially,
$\varpi>Q$ on $\Sigma$. 
We can already state in particular:
\begin{corollary}
\label{stabtoR}
If our initial data as in Proposition~\ref{frameprop}
is assumed sufficiently close to the data of a given subextremal Reissner--Nordstr\"om 
solution,
then $(\ref{nonex})$ is indeed satisfied and the conclusions of the above theorem hold.
\end{corollary}

\section{$\mathcal{CH}^+_{A,B}$ is (globally) singular}
\label{globsingsec}

Under an additional assumption on the behaviour of the scalar field
on $\mathcal{H}^+_{A,B}$, it is shown that $\mathcal{CH}^+_{A,B}$  in fact
correspond to 
\emph{weak null singularities}. 
These can be thought of as singular null hypersurfaces in 
a $C^0$ extension of the metric, singular in that the Kretschmann scalar
blows up, in fact, so that the
Christoffel symbols are singular in a sense to be discussed in Section~\ref{christoffels}.

Specifically, the following result was proven in~\cite{md:cbh}:
\begin{theorem}[\cite{md:cbh}]
\label{blowupth}
Under the assumptions of Theorem~\ref{chstab}, suppose
we have in addition that  
\begin{equation}
\label{lowerprice}
|\partial_v\phi|\ge cv^{-6+\tau  },      \qquad |\partial_u\phi|\ge cu^{-6+\tau}
\end{equation}
on $\mathcal{H}^+_A$, $\mathcal{H}^+_B$,
for all $v\ge V$, $u\ge U$ respectively, for some $U,V<\infty$, and some
$c>0$, $\tau>0$.
Then $\varpi$ extends ``continuously'' to $\infty$ on 
$\mathcal{CH}^+_{A,B}$ in a neighbourhood of $i^+_{A,B}$.

Moreover, choosing $\mathcal{U}_{A,B}$ in Theorem~\ref{chstab}
sufficiently small, then the curves $\mathcal{A}'_{A,B}\cap\mathcal{U}_{A,B}$
are in fact spacelike.
\end{theorem}

The reader should compare the required lower bounds $(\ref{lowerprice})$ 
with the upper bounds $(\ref{forprice'slaw})$
proven in Theorem~\ref{plt}. (In particular, $(\ref{lowerprice})$
is indeed compatible with $(\ref{forprice'slaw})$!). 
We have in fact the following
\begin{conjecture}
\label{isitgen}
For \emph{generic} data as in Thereom~\ref{chstab}, the assumptions of
Theorem~\ref{blowupth} hold.
\end{conjecture}
Resolution of the above would completely close the book
on the mass-inflation scenario in the context of Einstein--Maxwell--real
scalar field model, as the blow-up part of the scenario would then
be retrieved
for generic data posed on a spacelike asymptotically flat Cauchy surface.

We proceed in the remainder
of this section to  strengthen slightly Theorem~\ref{blowupth} to show that
$\varpi$ extends ``continuously'' to $\infty$ on the entire $\mathcal{CH}^+_{A,B}$.

First, note that in view of the second statement of Theorem~\ref{chstab},
we have that
there exists an open neighbourhood $\mathcal{U}_{A,B}$ of $i^+_{A,B}$ 
such that
\[
\mathcal{CH}^+_{A,B}\cap \mathcal{U}_{A,B}\subset \overline{\mathcal{T}}.
\]
We have in fact that this is true globally:
\begin{proposition}
\label{interestingly}
Under the assumptions of Theorem~\ref{blowupth},
\[
\mathcal{CH}^+_{A,B} \subset \overline{\mathcal{T}}.
\]
\end{proposition}

\begin{proof}
Define $r_0= \min_{\Sigma} r$. We have
by the Raychaudhuri equations $(\ref{Raych1})$, $(\ref{Raych2})$ that
 that if $(u,v)\in \mathcal{Q}$ for
which $r(u,v)<r_0$, then $(u,v)\in \mathcal{T}$.

Now, let $(u_1,v_1)\in \mathcal{CH}^+_A$ (WLOG, we shall prove the statement
for $\mathcal{CH}^+_A$)
for which $(u_1,v_1)\not\in \overline{\mathcal{T}}$.
It follows from the second statement of Theorem~\ref{chstab} 
that for each $\tilde{v}<v_1$ sufficiently close to $v_1$, there exists a $\tilde{u}(\tilde{v})$
such that $\partial_vr (\tilde{u},\tilde{v})=0 $, and such that $(u,\tilde{v})\in \mathcal{T}$
for $u^*<u<\tilde{u}$, where $(u^*,\tilde{v})\in \mathcal{A}'_A$.
Moreover, 
there exists an $\epsilon$ such that $\tilde{u}\ge u(\mathcal{H}^+)+\epsilon$
for all $\tilde{v}$. 

Since $\partial_vr (\tilde{u},\tilde{v})=0 $ we have
$r(\tilde{u},\tilde{v})=\varpi\pm \sqrt{\varpi^2-Q^2}$.
On the other hand, by the monotonicity   and Theorem~\ref{blowupth} above, it follows
that $\varpi(\tilde{u},\tilde{v})\to \infty$ as $\tilde{v}\to v_1$.
It follows that for $\tilde{v}$ sufficiently close to $v_1$, we have
 $r(\tilde{u},\tilde{v})=\varpi- \sqrt{\varpi^2-Q^2}\to 0$ as $\tilde{v}\to v_1$. 
In particular, for $\tilde{v}$ sufficiently close to $v_1$, $r(\tilde{u},\tilde{v})< r_0$
and thus  $(\tilde{u},\tilde{v})\in \mathcal{T}$,
a contradiction.
\end{proof}

\begin{corollary}
\label{globsingc}
Under the assumptions of   Theorem~\ref{blowupth}, 
$\varpi$ ``continuously'' extends to $\infty$ on the whole of
$\mathcal{CH}^+_{A,B}$.
\end{corollary}
\begin{proof}
We apply Theorem~\ref{blowupth}, the above proposition and again the monotonicity
$(\ref{monintrap})$.
\end{proof}

As an additional corollary we have the following
\begin{corollary}
Under the assumptions of   Theorem~\ref{blowupth}, 
$(\mathcal{M},g)$ is future inextendible as a manifold with $C^2$ Lorentzian metric.
\end{corollary}
\begin{proof}
Let $(\widetilde{\mathcal{M}}, \widetilde{g})$ be a non-trivial $C^2$ extension.
We apply Part III~of Proposition~\ref{frameprop}. Let $\gamma$ be such that
$\overline{\pi_{\mathcal{Q}}(\gamma|_{\mathcal{M}})}\cap \mathcal{CH}^+\ne\emptyset$. 
Since $\varpi$ extends ``continuously'' to $\infty$ on $\mathcal{CH}^+$, it follows
from $(\ref{Kretsch})$
that the supremum of the Kretschmann scalar on
$\gamma|_{\mathcal{M}}$ is $\infty$,
a contradiction.  
\end{proof}

The above extendibility statement can be strengthened,
but such statements below $C^2$ need care.
We will examine the consequence of the blow-up of mass
in our setting in Section~\ref{christoffels}. 
Cf.~the comments following Theorem~\ref{introtheo}.

Conjecture~\ref{isitgen} notwithstanding, 
it is an open question whether Proposition \ref{interestingly}
holds under the assumptions of  only Theorem~\ref{chstab}.

\section{Cauchy stability}
\label{CSsec}
We will apply a Cauchy stability argument in the following form.
Let us fix a subextremal Reissner--Nordstr\"om solution with parameters
$0<Q_{RN}<M_{RN}$, and an arbitrary spherically symmetric (but \emph{not necessarily
admissible}) Cauchy hypersurface $\Sigma$. We fix the gauge 
described before Theorem~\ref{chstab}.

We shall consider ``$\delta$-perturbations'' of the data on $\Sigma$.
A $\delta$-perturbation will  be such that (with respect to the above
gauge), the following pointwise inequalities hold on $\Sigma$:
\begin{equation}
\label{eps1}
|r-r_{RN}|<\delta,\qquad |\Omega-\Omega_{RN}|<\delta,
\end{equation}
\begin{equation}
\label{eps2}
|\partial_u r-\partial_u r_{RN}|<\delta, \qquad |\partial_v r-\partial_vr_{RN}|<\delta,
\end{equation}
\begin{equation}
\label{eps3}
|\partial_u \Omega- \partial_u \Omega_{RN}|<\delta, 
\qquad |\partial_v \Omega- \partial_v \Omega_{RN}|<\delta,
\end{equation}
\begin{equation}
\label{eps4-5}
|Q-Q_{RN}|<\delta, \qquad
|\varpi-M_{RN}|<\delta,
\end{equation}
\begin{equation}
\label{eps6}
|\partial_u\phi|<\delta, \qquad |\partial_v\phi|<\delta.
\end{equation}

Note that in the two-end case, our equations admit a local well-posedness
theory in $C^1$ of $(\Omega, r, \phi)$.
Our Cauchy stability statement is then as follows:
\begin{proposition}
\label{CSprop}
Given any $0<U,V<\infty$ and $\epsilon>0$ 
there exists a $\delta$
so that for  $\delta$-perturbations of subextremal Reissner--Nordstr\"om data,
the maximal Cauchy development $(\mathcal{M},g, F, \phi)$ 
contains the compact null rectangle
$[-U,U]\times[-V,V]$ and the estimates $(\ref{eps1})$--$(\ref{eps6})$ hold there, with 
$\epsilon$ in place of $\delta$.
\end{proposition}

In particular, we have,

\begin{corollary}
\label{admississue}
Let $\epsilon>0$.
Consider an arbitrary data set which
is a $\delta$-perturbation of Reissner--Nordstr\"om data
with parameters $M_{RN}$, $Q_{RN}$, and let $(\mathcal{M},g,  F,\phi)$
be the maximal Cauchy development.
Then for $\delta$ sufficiently small there exist past and future admissible Cauchy hypersurfaces $\Sigma_\pm\subset
\mathcal{M}$,
which coincide with $\Sigma$ for large values of $|x|$,
and such that the induced data on $\Sigma_\pm$ are $\epsilon$-perturbations
of Reissner--Nordstr\"om data on $\Sigma_\pm$, with respect to the normalisation
 described before Theorem~\ref{chstab}, now applied to Reissner--Nordstr\"om data
on $\Sigma_\pm$.
\end{corollary}

In view of the above corollary, we are justified in having restricted up to now
attention on the case $\Sigma=\Sigma_+$, and we shall always restrict
our attention in what follows to the case $\Sigma=\Sigma_+$, and to future
evolution,
in the proof of Theorem~\ref{introtheo}.

\section{A uniform estimate on the future boundary of the `stable blue-shift' region}

To prove the main theorem of the present paper,
we  appeal directly to another 
statement from~\cite{md:cbh}, which is used
in the proof of relation $(\ref{thenonemp})$ of Theorem~\ref{chstab}. 

\begin{proposition}[see Proposition 9.1  of \cite{md:cbh}]
\label{stableblue}
Under the assumptions of Theorem~\ref{chstab}, there exists a spacelike curve $\gamma_A$ terminating at $i^+_A$,
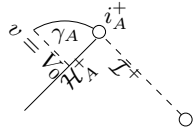
\[
\input{glob_gam.pstex_t}
\]
an $s>0$, a $C_s>0$,
such that on $\gamma_A$,
\begin{equation}
\label{propagating}
0< -\partial_v r  \le C_s v^{-1-s}
\end{equation}
and such that given any $\epsilon>0$, there exists a $V_0$ such that
\begin{equation}
\label{wellintheblue}
\left|r- \varpi^+_A+\sqrt{(\varpi^+_A)^2-Q^2}\right|\le \epsilon
\end{equation}
on $\gamma_A\cap\{v\ge V_0\}$. 
\end{proposition}

Examining the proofs of~\cite{md:cbh, dr1}, we easily see
\begin{claim}
Given $\epsilon_0>0$ and $C>0$, then 
there exists a $\delta_0>0$ such
that for all  $\delta_0$-perturbations of a fixed sub-extremal
Reissner--Nordstr\"om data-set with parameters $0<Q_{RN}<M_{RN}$
on an admissible $\Sigma$,
and such that moreover $(\ref{stoapeiro})$ holds with constant $C$, then
$(\ref{propagating})$ holds on $\gamma_A$ where 
$C_s$, $s$ can be chosen uniformly,  and such 
that $(\ref{wellintheblue})$ holds with $\epsilon=\epsilon_0$
for a $V_0$ that can be chosen  uniformly.
\end{claim}

Note that the position of $\gamma_A$ in $(u,v)$ coordinates depends
of course on the solution.

Finally, we remark that an identical construction gives us $\gamma_B$,
analogous to $\gamma_A$, with $u$ replacing $v$.

\begin{remark}
Let us note here that $\gamma_A$ above represents, in the terminology
introduced in~\cite{md:cbh}, the future boundary of the ``stable blue shift region''
$\mathcal{B}_\gamma$. See Section~6 of~\cite{md:cbh} for a discussion 
of these regions.
 As we shall see, the significance of
$(\ref{propagating})$ is that it is integrable in $v$, and this property will persist
to the future of $\gamma_A$.
\end{remark}

In the course of the proof of our main theorem, we will not in fact appeal to 
Theorem~\ref{chstab} or Corollary~\ref{stabtoR},  
but rather directly to Proposition~\ref{stableblue}.
In particular, we shall effectively recover the proof of the non-emptiness
relation $(\ref{thenonemp})$, given of course Proposition~\ref{stableblue}.

\section{$\mathcal{S}\cup\mathcal{S}_{A}\cup \mathcal{S}_B=\emptyset$}
Having recalled all that we will need, we may now prove the main result
of the paper. Given the framework of Proposition~\ref{frameprop} and the comment
following Corollary~\ref{admississue}, then the first
two statements of 
Theorem~\ref{introtheo} of the introduction will follow from
\begin{theorem}
\label{willfollowf}
Let $0<Q_{RN}<M_{RN}$. Let $C>0$ be a constant.
Then there exists a $\delta>0$  such that for
all $\delta$-perturbations
of subextremal Reissner--Nordstr\"om data on a future admissible $\Sigma$
with parameters $Q_{RN}$, $M_{RN}$,
satisfying moreover $(\ref{stoapeiro})$,
then 
\[
\mathcal{S}\cup\mathcal{S}_{A}\cup \mathcal{S}_B=\emptyset.
\]
Moreover, given $\epsilon>0$, then the global bound 
\[
r\ge M_{RN}-\sqrt{M_{RN}^2-Q_{RN}^2}-\epsilon
\]
holds, for $\delta$ chosen small enough.
\end{theorem}
\begin{proof}
Given $\epsilon>0$ and $C$,  let $0<\epsilon_0\ll \epsilon$ (to be determined),
and
let $\delta_0$, $V_0$, $s$, $C_s$ be as in the claim following
Proposition~\ref{stableblue}. 
We have
\begin{equation}
\label{giatoQ}
|Q-Q_{RN}|<\epsilon_0.
\end{equation}
We can then choose $V\ge \max\{V_0,-v(\mathcal{H}^+_B),0\}$ such that 
\begin{equation}
\label{integratedsmall}
\int_{V}^v 2C_sv^{-1-s}<\epsilon_0,
\end{equation}
and 
\begin{equation}
\label{fromthebound}
\left|r(u,v)- (M_{RN}-\sqrt{M^2_{RN}-Q^2_{RN}})\right| < \epsilon_0
\end{equation}
for $(u,v)\in \gamma_A$ and $v\ge V$.

We apply the same argument to the end $B$ and this gives us a $U_0$, a new $\delta_0$,
and a $\gamma_B$. Let us call $\delta_0$ the minimum of the new $\delta_0$ and
the previous.
We may then choose a $U\ge \max\{U_0, -u(\mathcal{H}^+_A),0)\}$ satisfying the analogue
of $(\ref{integratedsmall})$ with $u$, $U$, in place of $v$, $V$.

Let us note that in the Reissner--Nordstr\"om solution with parameters
$M_{RN}$, $Q_{RN}$, we have that 
\[
r_{RN}\ge M_{RN} -\sqrt{M^2_{RN}-Q^2_{RN}}+\epsilon_{UV}
\]
on the rectangle $[-U,U]\times [-V,V]$,
while
\[
r_{RN} \le M_{RN}-\sqrt{M^2_{RN}-Q^2_{RN}} +\epsilon_0
\]
on $[u_{A,RN}(V),U]\times\{V\}$, $\{U\}\times[v_{B,RN}(U),V]$,
where by definition, 
$(u_{A,RN}(V),V)\in \gamma_{A,RN}$, 
$(U,v_{B,RN}(U)) \in\gamma_{B,RN}$.

Now given these $U,V$ and choosing  $\epsilon_1\ll\epsilon_{UV}, \epsilon_0$, we
apply  Proposition~\ref{CSprop} (with $\epsilon_1$ in place of $\epsilon$), 
and let $\delta$ be the minumum
of the previous $\delta_0$ and that given by the above corollary.

For $\epsilon_1\ll\epsilon_{UV}, \epsilon_0$ sufficiently small, it follows
that on $[u_{A}(V),U]\times\{V\}$ and $\{U\}\times[v_{B}(U),V]$, we have
\begin{equation}
\label{suv9nkn}
r \ge M_{RN}-\sqrt{M^2_{RN}-Q^2_{RN}} - \epsilon_0,
\end{equation}
and also
\[
\varpi>Q,
\]
\[
\varpi-\sqrt{\varpi^2-Q^2}<r< \varpi+\sqrt{\varpi^2-Q^2},
\]
and thus, $1-2m/r<0$.
In view of the fact that the signs of $\partial_ur$, $\partial_vr$, respectively are
determined by Raychaudhuri on these sets, it 
follows that $[u_{A,RN}(V),U]\times\{V\}\cup \{U\}\times[v_{B,RN}(U),V]
\subset\mathcal{T}$. In view of this together with the first inequality of
$(\ref{propagating})$ and its $u$-analogue, we obtain
\begin{equation}
\label{TOMELLOV}
J^+(\gamma_A)\cup J^+(\gamma_B) \cap (\{v\ge V\}\cup \{u \ge U\})\subset \mathcal{T}.
\end{equation}

We will show that
$(\ref{fromthebound})$   holds on the region $(\ref{TOMELLOV})$
with $\epsilon_0$ replaced by $3\epsilon_0$.

For this, let us define the region
\[
\mathcal{X}= J^+(\gamma_A)\cup J^+(\gamma_B)\, \cap\, (\{v\ge V\}\cup \{u \ge U\})
\cap \{r \ge M_{RN}-\sqrt{M^2_{RN}-Q^2_{RN}}-4\epsilon_0 \}
\]
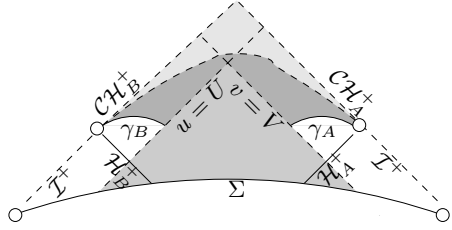
\[
\input{glob_boot.pstex_t}
\]

In view of~$(\ref{suv9nkn})$, $\mathcal{X}\supset
[u_{A}(V),U]\times\{V\} \cup \{U\}\times[v_{B}(U),V]$,
while in view of~$(\ref{fromthebound})$, $\mathcal{X}\supset
(\gamma_A\cap\{v\ge V\})\,\cup\, ( \gamma_B\cap\{u\ge U\})$.

By $(\ref{TOMELLOV})$, we have that $\mathcal{X}\subset\mathcal{T}$.
We have thus the one-sided bound
\begin{equation}
\label{onesided}
r \le M_{RN}-\sqrt{M^2_{RN}-Q^2_{RN}}+\epsilon_0.
\end{equation}
Let us note that by the monotonicity properties of $r$, then
$\mathcal{X}$ is a past set in the topology of the left hand side of $(\ref{TOMELLOV})$.

Now choosing $\epsilon$ sufficiently small,
then the inequality $M_{RN}>Q_{RN}$, the definition of $\mathcal{X}$ and
the bound $(\ref{onesided})$ yields the one-sided bound
\begin{equation}
\label{tobasiko}
-\frac{1}{4r}+\frac{Q^2}{4r^3} >0
\end{equation}
in $\mathcal{X}$, and
thus
by $(\ref{requ})$ we have
\[
\partial_u\log (-\partial_v r) \le  \frac1{r} (-\partial_ur).
\]

Integrating in $u$ from $\gamma_A(v)$, we obtain
in the region  $\mathcal{X}\cap \{v\ge V\}$ the bound
\[
\log (-\partial_v r)(u,v) < E \epsilon_0 + \log (-\partial_v r)(u_A(v),v)
\]
where $(u_A(v),v)\in \gamma_A$,
and $E$ is a constant depending on $M_{RN}$, $Q_{RN}$.
Choosing $\epsilon_0$ so that $e^{E\epsilon_0}\le 2$, we obtain
\begin{equation}
\label{itpropagated}
0<-\partial_v r < 2 C_s v^{-1-s}.
\end{equation}

One obtains similarly to $(\ref{itpropagated})$,
\begin{equation}
\label{itpropagated2}
0<-\partial_u r < 2C_s u^{-1-s}
\end{equation}
in  $\mathcal{X}\cap \{u\ge U\}$.

Now  integrating $(\ref{itpropagated})$
from $[u_A,U]\times \{v=V\} \cup (\gamma_A\cap \{v\ge V\})$ 
one obtains in view of $(\ref{integratedsmall})$ that 
\[
r\ge     M_{RN}-\sqrt{M_{RN}^2-Q^2_{RN}}  - 2\epsilon_0
\]
in 
$\mathcal{X}\cap \{u\le U\}$.
It follows that
\[
\mathcal{X}\supset \cup_{v\ge V}[(u_A(v),U] \times\{v\},
\]
in particular
\[
\mathcal{X}\supset \{U\}\times [V,\infty).
\]

We may now interate $(\ref{itpropagated2})$ in $u$ from
$\{U\}\times[v_B,\infty)$ to obtain  similarly
\begin{equation}
\label{finalbound}
r\ge     M_{RN}-\sqrt{M_{RN}^2-Q^2_{RN}}  - 3\epsilon_0
\end{equation}
in 
$\mathcal{X}\cap \{u\ge U\}$.
This now implies that
\[
\mathcal{X}=
J^+(\gamma_A)\cup J^+(\gamma_B) \cap (\{v\ge V\}\cup \{u \ge U\})
\]
with the bound $(\ref{finalbound})$ holding on all of $\mathcal{Q}$.
The theorem follows, choosing  (in addition to the above restrictions) $3\epsilon_0<\epsilon$.
\end{proof}

\section{The continuous extension}
Let $u'>u(\mathcal{H}^+_A)$, $v'>v(\mathcal{H}^+_B)$.
We shall define a new set of null coordinates $(u^*,v^*)$ centred
at $(u',v')$, 
i.e.~such that $(u^*,v^*)=(0,0)$ at the old $(u',v')$,
and such that with respect to the new coordinates,
$(\Omega^*)^2=1$ on $u^*=0$, $v^*=0$.
{\bf Let us however immediately
drop the $*$ and refer to these new coordinates again as $(u,v)$.}

Since $\partial_vr $ becomes negative on $v=0$ and $\partial_u r$ becomes negative
on $u=0$, then from $(\ref{Raych1})$, $(\ref{Raych2})$ and the result of the
previous section, it follows
that $J^+(\mathcal{H}^+_A) \cup  J^+(\mathcal{H}^+_B)$
corresponds to a \emph{finite}
range
\[
(-U',U'')\times (-V',V'')
\]
for some $0<U', U''<\infty$, $0<V', V''<\infty$.
\begin{proposition}
Under the assumptions of Theorem~\ref{willfollowf},
then with respect to the new coordinate system defined above,
$\Omega^2$, $r$ and $\phi$ extend to continuous functions
on the half-closed rectangle
\[
(-U',U'']\times(-V',V''].
\]
Moreover, $\Omega$, $r$ and $\phi$ restricted
to both
$\{U''\}\times (-V',V'')$ and $(-U',U'')\times\{V''\}$ are smooth functions.
\end{proposition}

\begin{proof}
We begin with $r$.
The estimate of  the previous section   shows that  for $-U'<u<U''$
\[
\int _{v_1}^{v_2} (-\partial_vr)(u,v) dv \to 0
\]
as $v_1<v_2$ and $v_1\to V''$.
On the other hand,
the argument of the previous section also yields
the estimate 
\[
\sup_{v\ge 0} -\partial_ur(u,v)  \le -C \partial_u r(u,V''-\epsilon)
\]
for $\epsilon$ sufficiently small.
Since moreover
\[
\lim_{\tilde{u}\to U''}\int_{\tilde{u}}^{U''} -\partial_u r (u,V''-\epsilon)d u \to 0
\]
the required statement about $r$ follows immediately,
interchanging also the role of $u$ and $v$.

Note in fact the bound
\begin{equation}
\label{noteinfact}
\int_0^{U''}\int_{0}^{V''} |\partial_ur\partial_vr| \,du\,dv \le  C.
\end{equation}

We turn now to $\phi$.
Recall from   Section~13 of~\cite{md:cbh}, 
that, in view of the global lower bound on $r$,
and the bounds\footnote{note these are independent of the normalisation
of the null coordinates}
\[
\int_{\mathcal{H}^+_A} |\partial_v\phi| dv\le C, \qquad
\int_{\mathcal{H}^+_B} |\partial_u\phi| du\le C
\]
which follow from Theorem~\ref{plt},
one can obtain global uniform bounds 
\begin{equation}
\label{matterunif}
\int_{-V'}^{V''} \sup_{u\ge-U'}|\partial_v\phi|dv \le C, \qquad
 \int_{-U'}^{U''}\sup_{v\ge-V'}|\partial_u\phi| du \le C.
\end{equation}
Note that these bounds imply the global bound
\[
|\phi|\le C.
\]

Integrating the wave equation $(\ref{WAVe})$,
our estimates show that for each $v\in(-V',V'')$, and $-U'<u_0$, $-V'<v_0$,
we have
\[
\sup_{u\ge u_0} |\partial_v\phi(u,v)|\le C(v,u_0)
\]
and similarly
\[
\sup_{v\ge v_0}|\partial_u\phi(u,v)|\le C(u, v_0)
\]
for each $u\in(-U',U'')$, for functions $C(v,u_0)$, $C(u,v_0)$ such that
\begin{equation}
\label{integrableinuv}
\int_{u_0}^{U''} C(u,v_0)du <\infty, \qquad \int_{v_0}^{V''} C(v,u_0)dv<\infty.
\end{equation}
The continuity of $\phi$ now follows.

We finally turn to $\Omega$. 
From $(\ref{onesided})$ we see, for sufficiently small $\epsilon$, the one-sided bound
in the region $[0,U'')\times[V''-\epsilon,V'')\cup [U''-\epsilon,U'')\times[0,V'')$:
\begin{equation}
\label{twraauto}
\partial_u\partial_v \log \Omega^2 \le -4\pi \partial_u\phi\partial_v\phi
+\frac1{r^2}\partial_vr \partial_ur.
\end{equation}

From $(\ref{matterunif})$, it follows that the first term on
the right hand side of $(\ref{twraauto})$ is bounded
upon integration in absolute value over $[-U',U'']\times[-V',V'']$, 
while  from $(\ref{noteinfact})$, 
the second term on the right hand side of $(\ref{twraauto})$
is bounded upon integration in absolute value over $[0,U'']\times [0,V'']$.
Thus, these terms remain bounded when integrated over subregions.

Integrating now 
$(\ref{Omegeq})$
in $[0,u]\times [0,v]$, for $u\ge0$, $v\ge 0$,
using the compactness of $[0,U''-\epsilon]\times[0,V''-\epsilon]$,
the inequality
$(\ref{twraauto})$ in the remaining region, and the bounds of
the preceding paragraph, 
we obtain thus the uniform one sided bound
\[
\log\Omega^2 \le C.
\]

In particular, we have the global bound 
\begin{equation}
\label{Cinv}
\int_0^{U''}\int_0^{V''} \Omega^2 dudv <\infty.
\end{equation}

Now let us note that
the above argument would have applied if we chose $u=0$, $v=0$ at an earlier
time (always strictly to the future of the  two event horizons), and
$(\ref{Cinv})$ is a coordinate invariant statement!

It follows that
for any $-U'<u_0<U''$, $-V'<v_0<V''$,
we have
\begin{equation}
\label{fromCinv}
\int_{u_0}^{U''}\int_{v_0}^{V''} \Omega^2 dudv <C(u_0,v_0).
\end{equation}

Using $(\ref{fromCinv})$ and again $(\ref{matterunif})$, we may now 
integrate $(\ref{Omegeq})$
backwards and forwards from $u=0$, $v=0$, to 
obtain, in $[u_0,U'')\times[v_0,V'')$, a two-sided bound
\[
|\log\Omega^2| \le C(u_0,v_0).
\]

Revisiting the equation $(\ref{Omegeq})$, we obtain similar bounds
\[
\sup_{u\ge u_0} |\partial_v \Omega|(u,v)\le C(v,u_0),
\]
\[
\sup_{v\ge v_0} 
|\partial_u \Omega|(u,v)\le C(u,v_0),
\]
where $C(u,v_0)$, $C(v,u_0)$ again satisfy $(\ref{integrableinuv})$.

The continuous extendibility of $\Omega$ follows.

We have in fact already proven more, namely, that, in addition
to continuous extendibility,
$r$, $\phi$ satisfy global $BV$ bounds on both
$(-U',U'']\times\{V''\}$ and $\{U''\}\times(-V',V'']$,
while $\Omega$ satisfies local $BV$ bounds on these sets.

The higher order regularity claimed by the proposition follows inductively
by differentiating the basic equations and applying the above bounds.
\end{proof}

\section{Blow-up of Christoffel symbols}
\label{christoffels}

Given the above continuous extension, let us examine in 
more detail the geometry of the singular surface in the case where
the assumptions of Theorem~\ref{blowupth} apply.
We assume thus that in addition to the assumptions of Theorem~\ref{willfollowf},
we have $(\ref{lowerprice})$, 
and consequently by Corollary~\ref{globsingc}, $\varpi$ extends
``continuously'' to $\infty$ identically on $\mathcal{CH}^+$ (let us note
that the monotonicity shows that we have ``continuity'' also at the bifurcation point).

With respect to the finite coordinates defined in the last section, which
are to be thought now as a regular coordinate system of a $C^1$ (in fact
$C^\infty$) 
manifold admitting a $C^0$
continuous extension $\widetilde{g}$ of $g$, 
we see from the definition of Hawking mass~$(\ref{mdef})$ and
the bounds on $\Omega$, $\partial_ur$ that for $-U'<u<U''$,
\[
\lim_{v\to V''}\partial_v r(u,v) \to-\infty
\]
and (in view of the bounds on $\partial_vr$)
\[
\lim_{u\to U''}\partial_u r(u,v)\to -\infty.
\]
In particular, the extension $\widetilde{g}$ cannot have $C^0$
Christoffel symbols in $u$, $v$ coordinates
(for instance,
$\Gamma^u_{AB}=2\Omega^{-2}r \partial_vr\sigma_{AB}$, 
with $A,B$ coordinates on $\mathbb S^2$).

From equations $(\ref{Raych1})$, $(\ref{Raych2})$, 
in view again of the  bounds on $\Omega$, it follows that
for $-U'<u<U''$
\[
\int_v^{V''}( \partial_v\phi)^2(u,v)dv =\infty
\]
while for $-V'<v<V''$
\[
\int_u^{U''}( \partial_u\phi)^2(u,v)du =\infty.
\]

It follows in particular that $\phi$ does not extend $H^1_{\rm loc}$ at any point
of $\mathcal{CH}^+$.

To continue, let us note first that
we know a bit more from~\cite{md:cbh}, namely that, under the assumptions
of Theorem~\ref{blowupth}, 
then
in a neighbourhood $\mathcal{U}_A$ of $i^+_A$, 
then WLOG, $\partial_u\phi >0$, $\partial_v \phi>0$
in $\mathcal{U}_A\cap  J^+(\gamma_A)$, and thus,
noting from $(\ref{WAVe})$ that it follows that
$\partial_v(r\partial_u\phi)>0$, we have that
$\partial_u\phi>0$ on $\mathcal{CH}^+_A\cap \mathcal{U}_A$.
(Let us note  moreover that this and our upper bounds on $\partial_u\phi$ 
imply that in fact $\partial_v\phi>0$ in a neighbourhood of $(-U',U'')\times\{V''\}$.)

Let $u_1<u_2$ be two $u$-values such that $u=u_1$ and $u=u_2$ intersect
$\mathcal{CH}^+_A\cap \mathcal{U}_A$. Note by the continuity properties of
$\partial_u\phi|_{\mathcal{CH}^+_A}$ and compactness, we have for
$v_0<V''$ suitable large, a lower bound
$|\partial_u\phi|\ge c(v_0, u_1,u_2)$ in $v\ge v_0$, $u_1\le u\le u_2$,
for a positive $c>0$.

One sees easily that 
\begin{eqnarray*}
\int_{v_0}^{V''} (\partial_v\Omega)^2(u_1,v)dv+\int_{v_0}^{V''}  (\partial_vr)^2(u_1,v)dv
+\int_{v_0}^{V''} (\partial_v\Omega)^2(u_2,v)dv\\
\ge c \int_{v_0}^{V''}(\partial_v\phi)^2(u_1,v) dv 
-C
\end{eqnarray*}
from which it follows that there exists a $u_1\le u_0\le u_2$
such that for $u\ne u_0$,
then 
\begin{equation}
\label{firstcase}
\int_{v_0}^{V''} (\partial_v\Omega)^2(u,v)dv=\infty
\end{equation}
or 
\begin{equation}
\label{secondcase}
\int_{v_0}^{V''} (\partial_vr)^2(u,v)dv=\infty.
\end{equation}

In particular, the Christoffel symbols fail to be $L^2_{\rm loc}$ at any
point of $\mathcal{CH}^+_A\cap\mathcal{U}_A$ (recall
$\Gamma^u_{AB}=2\Omega^{-2}r \partial_vr\sigma_{AB}$
and note also $\Gamma^v_{vv}=
\partial_v\log\Omega^2$).

Let us note that if $(\ref{secondcase})$ is true for a particular value of $u$,
then it is true for all $u\in(-U',U'']$.

On the other hand, if $(\ref{secondcase})$ does \emph{not} hold, then
we easily see the following: Let $Z=\{u:\partial_u\phi(u,V'')\ne 0\}$.
We have that for almost all $u\in Z$, $(\ref{firstcase})$ holds.

Similar considerations apply to $\{u=U''\}$.   

It follows thus that the Christoffel symbols fail to be $L^2_{\rm loc}$ at any
point $x$ of $\mathcal{CH}^+$ such that $\phi|_{\mathcal{CH}^+}$ is not
constant in a neighbourhood of $x$.

One expects that the statements shown in this section can be generalised to apply to any
nontrivial $C^0$ Lorentzian manifold $(\widetilde{\mathcal{M}},\widetilde{g})$ extending
$(\mathcal{M},g)$, not assumed \emph{a priori} to have the above differential structure.
We shall not pursue the
details here.

\section{Connection with the Luk--Rodnianski theory of interacting
impulsive gravitational waves}
\label{lukrodtheory}

In a different guise, ``weak null singularities''
have a long history in general relativity. They occur in 
the explicit colliding plane wave spacetimes constructed by Khan--Penrose~\cite{khanpen}
and Szekeres~\cite{szek}, where intersecting null hypersurfaces carry
delta-function singularites in the curvature.
In these spacetimes, however, the
singularities are even {\bf weaker} than the ``weak null singularities'' of
our Theorem~\ref{thes2}, 
in that the Christoffel symbols remain bounded. 
In view of this,  the  plane waves of~\cite{khanpen} were not interpreted as
terminal singular boundaries of spacetime, but
rather as singular hypersurfaces in an ambient spacetime
continuing on both sides of the hypersurface,
whose dynamics should still be uniquely prescribed
by the Einstein equations, interpreted in a weak
sense. I.e.~the situation should be much like that for
well known shock fronts of classical fluid dynamics, where
the equations uniquely determine the dynamics (after imposing in addition a suitable
entropy condition).

A mathematical theory incorporating the above (and more general!)
singular hypersurfaces has only been
achieved very recently.  The reason this problem is so difficult
is that this type of singularity, despite
being sufficiently ``weak'' so as to allow for the spacetimes
to be interpreted as weak
solutions of the Einstein equations across the singularity, 
is still below the general well-posedness regularity threshold
for the Einstein equations  (even after the recent  seminal
improvement due to Klainerman--Rodnianski--Szeftel~\cite{krs} requiring only that
the  metric be in 
the Sobolev space $H^2$ on a Cauchy surface).
In remarkable very recent work of 
Luk--Rodnianski~\cite{LukRod1}, a well-posedness
result has been achieved 
for the characteristic initial value problem
with characteristic data admitting a delta-function singularity
of the curvature across a $2$-sphere, but otherwise smooth. This singularity is then
 shown to propagate along a null hypersurface emanating from
the $2$-sphere. These null singular hypersurfaces
can be thought of as impulsive gravitational
waves, generalising the plane symmetric examples of
Khan--Penrose and imbedding them finally in a dynamical framework
for which uniqueness is understood.

 Moreover,
in~\cite{LukRod2}, Luk and Rodnianski show that two impulsive
gravitational waves interact so that     the singularity
is supported in the union of two null cones. In particular, 
a singularity does not form  before
their intersection, and their ``wake'' is completely regular!

Although as we have said, 
the null singularities of Luk--Rodnianski are less singular than 
the weak null singularities of the present paper (as they must be
to allow for a well posedness theorem for spacetimes continuing beyond them!),
it is impossible not to compare the  situation established in~\cite{LukRod2} with the picture
of  Theorem~\ref{introtheo}, where again two weak null singularities propagate and
interact (albeit now as the terminal boundary of spacetime) without a ``stronger''
singularity occuring first.\footnote{Moreover, one can strengthen the connection
with the following remark. Were one to drop the ``constraint'' equations 
$(\ref{Raych1})$--$(\ref{Raych2})$
and solve only the     ``evolution'' equations $(\ref{requ})$--$(\ref{WAVe})$, then
one may still reasonably interpret the evolution equations alone in a generalised sense 
\emph{across}
$\mathcal{CH}^+$, and extend the spacetime so as to satisfy these. The
hypersurfaces $\mathcal{CH}^+_{A,B}$ can then be thought of precisely
as interacting singularities in such a spacetime, which again would effectively
pass through one another.  It is interesting
that it is specifically the impossibility of satisfying the constraint equations that exclude
such extensions from being thought of as weak solutions of the Einstein equations. 
In this sense, the constraint
equations can be thought to play a role analogous to that
of an entropy condition in classical continuum physics.}
The
methods introduced by~\cite{LukRod1,LukRod2}
thus give hope that {\bf vacuum}
spacetimes with singular structure similar to that of Theorem~\ref{introtheo} (and
without underlying symmetries)
will soon be constructed, parameterised by free initial data. 
This will be an important
 first step toward understanding
the validity of Conjecture~\ref{introconj}.

\section{Epilogue: Cosmological spacetimes without spacelike singularities}
\label{epilogue}

As discussed before, the original expectation of ``BKL''-type behaviour inside black hole 
regions starts from the    observation that locally the interior of a Schwarzschild
black hole can be viewed as a cosmological region. This relation goes both ways,
however. In particular, perturbing the Reissner--Nordstr\"om--de Sitter solution,
similar considerations to Reissner--Nordstr\"om apply. In fact, as we  shall see,
the situation is even more interesting, and the conclusions for strong cosmic censorship
potentially ominous.

\subsection{A rough stability result}
Let us first state a basic stability result, an exact analogue of the
stability part of Theorem~\ref{introtheo}, which can be proven following
the methods of this paper and~\cite{md:cbh, dr1, dr0:bound, dafren},  
depends on a relatively rough understanding for the decay properties
of the scalar field between the event and cosmological horizons, and
holds for all subextremal values of the Reissner--Nordstr\"om--de Sitter solution.
\begin{theorem}
\label{canbeshown}
For sufficiently small spherically
symmetric perturbations of subextremal Reissner--Nordstr\"om--de Sitter
initial data (with topology $\mathbb S^1\times \mathbb S^2$  and $n$ black hole regions), 
the maximal future (past) Cauchy development has Penrose diagramme given 
by:
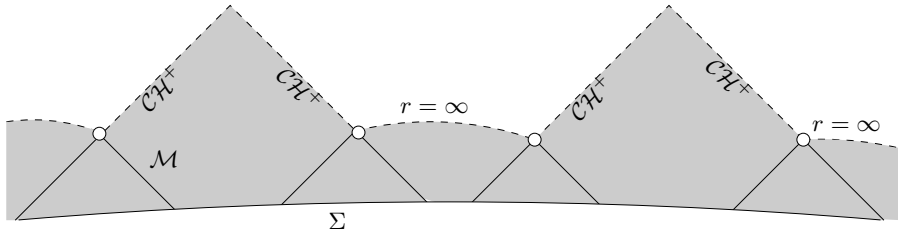
\[
\input{glob_cosmo.pstex_t}
\]
and the spacetime is  extendible as a manifold with $C^0$-Lorentzian
metric (to which $\phi$ also extends continuously), 
such that each connected component of $\mathcal{CH}^+$ is a 
bifurcate null cone in the extension.
\end{theorem}

The above already motivates a conjecture which is stated here in the Kerr case so as to
apply in fact for vacuum:
\begin{conjecture}
\label{cosmocon}
Let $(\mathcal{M},g)$ be the maximal Cauchy development (for the Einstein
vacuum equations with positive cosmological constant) of
sufficiently small perturbations of subextremal Kerr--de Sitter data (with $n$ black hole
regions) and initial 
topology ${\mathbb S}^1\times {\mathbb S}^2$. Then 
$(\mathcal{M},g)$ is future (past) extendible to a $(\widetilde{\mathcal{M}},\widetilde{g})$
with $C^0$ metric  $\widetilde{g}$, such 
that $\partial\mathcal{M}$ is the union of $n$ bifurcate null cones, and
all incomplete, future(past)-intextendible geodesics of $\mathcal{M}$ pass
into $\widetilde{\mathcal{M}}\setminus \mathcal{M}$.
\end{conjecture}
According to this conjecture, there is an {\bf open set in the moduli space} of vacuum,
``cosmological'' spacetimes
whose singularities, whether past or future,  are {\bf \emph{nowhere} spacelike}.

\subsection{The structure of the null singularity and the near extremal case}
The question of the precise generic singular nature of the null boundary
of these spacetimes is more delicate and depends on the parameters
 of the  Reissner--Nordstr\"om--de Sitter spacetime which is initially perturbed.
 
 This problem was plagued by much confusion and several incorrect claims.
The original studies of the charged 
spherical toy models gave a detailed analysis of the expected singular behaviour,
identifying a parameter range (i)
where the mass inflation scenario holds just as in the asymptotically flat case, 
a parameter range (ii) where  the curvature blows up without the mass blowing up, and,
finally, and most surprisingly, a parameter 
range (iii) where the curvature remains continuous.  See in particular~\cite{bradpoiss,
bradetal} and the survey~\cite{chambers}.
These original studies~\cite{bradpoiss, bradetal}, however, had failed to properly address the    
effect of backscattering in the region between the cosmological and event
horizons, already present in the scalar field model, which heuristically generates
an exponentially decaying tail with the slower exponential rate governed by the surface gravity
of the \emph{event horizon}, as opposed to the faster rate
governed by the surface gravity of the \emph{cosmological horizon}.
When this
was taken into account, the range of stability (iii) predicted in~\cite{bradpoiss, bradetal}
disappears. See~\cite{bradymossmyers}.
This suggested that the ``$C^2$-inextendibility formulation'' of strong cosmic censorship
holds in the case of positive cosmological constant--hence the reassuring
title of~\cite{bradymossmyers}.

Subsequent insights~\cite{formation} into the proper formulation of strong
cosmic censorship coming from the PDE point of view,
coupled with the recent Luk--Rodnianski
theory~\cite{LukRod1} described in the previous section, 
suggest that--\cite{bradymossmyers} not withstanding--one should
revisit the more detailed previous analyses of the studies described just above.
A revised heuristic analysis in the spirit of~\cite{bradetal} but
correctly taking into account the behaviour of the scalar field outside
the event horizon~\cite{bradymossmyers} suggests
that there is still a qualitative
difference in the nature of the blow-up, depending on the parameters,
corresponding to  (i) and (ii) above. 
In particular, close
to extremality, this suggests that although the curvature blows up,
 \emph{this is not associated with  mass blow up}, and,
particularly relevant to the point of view advanced here, the metric  in fact  extends Lipschitz
(its Christoffel symbols are uniformly bounded), and the scalar field extends
in $H^1_{\rm loc}$. 

\begin{conjecture}
In Theorem~\ref{canbeshown}, if the initial data is
sufficiently close to a  Reissner--Nordstr\"om--de Sitter 
with parameters in turn sufficiently close to extremal (but still subextremal!), then
the spacetime extensions $(\widetilde{\mathcal{M}},\widetilde{g})$ 
can be chosen so that $\widetilde{g}$ is Lipschitz, and 
such that the scalar field extends to $\widetilde{\mathcal{M}}$
as an $H^1_{\rm loc}$ function.

For generic such initial data (for all sub-extremal values of the Reissner--Nordstr\"om--de
Sitter reference solution), 
$\widetilde{g}$ cannot be $C^2$ and $\phi$ cannot extend $C^1$.
\end{conjecture}

In fact, the linearised analysis suggests
that the scalar field will extend in $W^{1,p}$,
where one can take $p\to \infty$  as  the initial data approaches extremality.
It is slightly amusing that the closer the black hole is to extremality,
the more ``stable'' its  Cauchy horizon, when set against the backdrop
of the novel instability
discovered by Aretakis~\cite{aretak1, aretak2} in his  study of the wave equation    
on exactly extreme Reissner--Nordstr\"om; this then suggests that,
as the parameters of the initial reference Reissner--Nordstr\"om--de Sitter
are pushed towards extremality, the size of the allowed ``open set''
in the moduli space of data, corresponding to the domain of stability,
must be taken to shrink.

When ``translated'' to the vacuum case (see the comments at the end of
Section~\ref{ssred}) one may expect that the $C^1$ blow up of the scalar field
will  be reflected as a $C^0$ blow up of the  Christoffel symbols,
and thus, the metric will not be Lipschitz. On the other hand, the above
still suggests that the  Christoffel symbols will be in $L^2_{\rm loc}$ (in fact,
in any $L^p_{\rm loc}$, given data sufficiently close to extremality).
This would represent a type of singularity which, although  outside the general
$H^2$-well posedness
theory of~\cite{krs}, may be describable by a suitable generalisation of
Luk--Rodnianski theory~\cite{LukRod1} outlined in the previous section, and
thus such solutions may
in fact admit     a local uniqueness characterization.  Thus, one is led to:
\begin{conjecture}
\label{cosmocon2}
Let $(\mathcal{M},g)$ be the maximal Cauchy development of
sufficiently small perturbations of Kerr--de Sitter data (with $n$ black hole
regions) and initial 
topology ${\mathbb S}^1\times {\mathbb S}^2$ with parameters sufficiently close to 
extremality. Then one can choose extensions $(\widetilde{\mathcal{M}},\widetilde{g})$ 
as above such 
that $\widetilde{g}$ has locally square integrable Christoffel symbols
and such that moreover
the extensions satisfy the Einstein equations in
a generalised sense.
In particular, the ``inextendible as a metric with
$L^2_{\rm loc}$ Christoffel symbols'' formulation (see~\cite{formation}) 
of strong cosmic censorship
does not hold in the case of positive cosmological constant.
\end{conjecture}

From a sufficiently coarse point of view, the spacetimes $(\mathcal{M},g)$ described
by the above conjecture would be stable
examples of {\bf cosmological black hole spacetimes without singularities--full stop}.

\subsection{Is strong cosmic censorship false?}
It may be tempting, in view of Conjecture~\ref{cosmocon2}, to prefer the $C^2$  formulation  
of strong cosmic censorship, and to invoke the pointwise blow-up of curvature as an indication
that the effects of quantum gravity will somehow save the day. Until that time where
such considerations can be well formulated, however,
the local PDE-properties governing the Cauchy problem for the Einstein equations provide
the best insight into the theory that we have,
and its conclusions--like them or not--cannot be taken lightly.
Thus, these considerations suggest that while the uniqueness property
captured by strong cosmic censorship
may    indeed hold in the      realm of   gravitational collapse, 
its validity as a fundamental physical principle
now  appears in doubt.

\section{Acknowledgements}
This paper was inspired by a conversation with Amos~Ori  at the ``Numerical relativity beyond astrophysics workshop'' (hosted by the International Centre for Mathematical
Sciences, Edinburgh, July 2011), whom the author also thanks for his
subsequent input which has led to several improvements, particularly with
regard to the discussion of the cosmological case.
The author  thanks the organisers of the workshop (David Garfinkle, Carsten Gundlach
and Luis Lehner) for facilitating a stimulating discussion session on null singularities.  
Additional thanks go to Igor Rodnianski,
Stefanos Aretakis, as well as Jonathans Kommemi and Luk.
The author is supported in part by a grant from the European Research Council.

\end{document}

%% file: glob_RN.pstex_t
\begin{picture}(0,0)%
\includegraphics{glob_RN.pstex}%
\end{picture}%
\setlength{\unitlength}{2960sp}%
\begingroup\makeatletter\ifx\SetFigFont\undefined%
\gdef\SetFigFont#1#2#3#4#5{%
  \reset@font\fontsize{#1}{#2pt}%
  \fontfamily{#3}\fontseries{#4}\fontshape{#5}%
  \selectfont}%
\fi\endgroup%
\begin{picture}(3703,2248)(2885,-7397)
\put(4051,-6436){\rotatebox{45.0}{\makebox(0,0)[lb]{\smash{{\SetFigFont{9}{10.8}{\rmdefault}{\mddefault}{\updefault}{\color[rgb]{0,0,0}$\mathcal{CH}^+$}%
}}}}}
\put(4051,-6886){\makebox(0,0)[lb]{\smash{{\SetFigFont{9}{10.8}{\rmdefault}{\mddefault}{\updefault}{\color[rgb]{0,0,0}$\mathcal{M}$}%
}}}}
\put(3301,-7196){\rotatebox{45.0}{\makebox(0,0)[lb]{\smash{{\SetFigFont{9}{10.8}{\rmdefault}{\mddefault}{\updefault}{\color[rgb]{0,0,0}$\mathcal{I}^+$}%
}}}}}
\put(5881,-6926){\rotatebox{315.0}{\makebox(0,0)[lb]{\smash{{\SetFigFont{9}{10.8}{\rmdefault}{\mddefault}{\updefault}{\color[rgb]{0,0,0}$\mathcal{I}^+$}%
}}}}}
\put(4829,-5856){\rotatebox{315.0}{\makebox(0,0)[lb]{\smash{{\SetFigFont{9}{10.8}{\rmdefault}{\mddefault}{\updefault}{\color[rgb]{0,0,0}$\mathcal{CH}^+$}%
}}}}}
\put(5040,-5743){\makebox(0,0)[lb]{\smash{{\SetFigFont{9}{10.8}{\rmdefault}{\mddefault}{\updefault}{\color[rgb]{0,0,0}$\gamma$}%
}}}}
\put(4216,-5881){\makebox(0,0)[lb]{\smash{{\SetFigFont{9}{10.8}{\rmdefault}{\mddefault}{\updefault}{\color[rgb]{0,0,0}$\widetilde{\mathcal{M}}$}%
}}}}
\put(4761,-7276){\makebox(0,0)[lb]{\smash{{\SetFigFont{9}{10.8}{\rmdefault}{\mddefault}{\updefault}{\color[rgb]{0,0,0}$\Sigma$}%
}}}}
\end{picture}%

%% file: glob_SCH.pstex_t
\begin{picture}(0,0)%
\includegraphics{glob_SCH.pstex}%
\end{picture}%
\setlength{\unitlength}{2960sp}%
\begingroup\makeatletter\ifx\SetFigFont\undefined%
\gdef\SetFigFont#1#2#3#4#5{%
  \reset@font\fontsize{#1}{#2pt}%
  \fontfamily{#3}\fontseries{#4}\fontshape{#5}%
  \selectfont}%
\fi\endgroup%
\begin{picture}(3703,1097)(2885,-7397)
\put(4051,-6886){\makebox(0,0)[lb]{\smash{{\SetFigFont{9}{10.8}{\rmdefault}{\mddefault}{\updefault}{\color[rgb]{0,0,0}$\mathcal{M}$}%
}}}}
\put(3301,-7196){\rotatebox{45.0}{\makebox(0,0)[lb]{\smash{{\SetFigFont{9}{10.8}{\rmdefault}{\mddefault}{\updefault}{\color[rgb]{0,0,0}$\mathcal{I}^+$}%
}}}}}
\put(5881,-6926){\rotatebox{315.0}{\makebox(0,0)[lb]{\smash{{\SetFigFont{9}{10.8}{\rmdefault}{\mddefault}{\updefault}{\color[rgb]{0,0,0}$\mathcal{I}^+$}%
}}}}}
\put(4771,-6716){\makebox(0,0)[lb]{\smash{{\SetFigFont{9}{10.8}{\rmdefault}{\mddefault}{\updefault}{\color[rgb]{0,0,0}$\gamma$}%
}}}}
\put(4146,-7296){\makebox(0,0)[lb]{\smash{{\SetFigFont{9}{10.8}{\rmdefault}{\mddefault}{\updefault}{\color[rgb]{0,0,0}$\Sigma$}%
}}}}
\end{picture}%

%% file: glob_pic.pstex_t
\begin{picture}(0,0)%
\includegraphics{glob_pic.pstex}%
\end{picture}%
\setlength{\unitlength}{2960sp}%
\begingroup\makeatletter\ifx\SetFigFont\undefined%
\gdef\SetFigFont#1#2#3#4#5{%
  \reset@font\fontsize{#1}{#2pt}%
  \fontfamily{#3}\fontseries{#4}\fontshape{#5}%
  \selectfont}%
\fi\endgroup%
\begin{picture}(2659,1048)(4268,-7322)
\put(5251,-7036){\rotatebox{45.0}{\makebox(0,0)[lb]{\smash{{\SetFigFont{9}{10.8}{\rmdefault}{\mddefault}{\updefault}{\color[rgb]{0,0,0}$\mathcal{H}^+$}%
}}}}}
\put(5851,-6886){\rotatebox{315.0}{\makebox(0,0)[lb]{\smash{{\SetFigFont{9}{10.8}{\rmdefault}{\mddefault}{\updefault}{\color[rgb]{0,0,0}$\mathcal{I}^+$}%
}}}}}
\put(5331,-6380){\rotatebox{315.0}{\makebox(0,0)[lb]{\smash{{\SetFigFont{9}{10.8}{\rmdefault}{\mddefault}{\updefault}{\color[rgb]{0,0,0}$\mathcal{CH}^+$}%
}}}}}
\put(5176,-7261){\makebox(0,0)[lb]{\smash{{\SetFigFont{9}{10.8}{\rmdefault}{\mddefault}{\updefault}{\color[rgb]{0,0,0}$\Sigma$}%
}}}}
\put(5898,-6584){\makebox(0,0)[lb]{\smash{{\SetFigFont{9}{10.8}{\rmdefault}{\mddefault}{\updefault}{\color[rgb]{0,0,0}$i^+$}%
}}}}
\put(6512,-7206){\makebox(0,0)[lb]{\smash{{\SetFigFont{9}{10.8}{\rmdefault}{\mddefault}{\updefault}{\color[rgb]{0,0,0}$i^0$}%
}}}}
\end{picture}%

%% file: glob_gen.pstex_t
\begin{picture}(0,0)%
\includegraphics{glob_gen.pstex}%
\end{picture}%
\setlength{\unitlength}{2960sp}%
\begingroup\makeatletter\ifx\SetFigFont\undefined%
\gdef\SetFigFont#1#2#3#4#5{%
  \reset@font\fontsize{#1}{#2pt}%
  \fontfamily{#3}\fontseries{#4}\fontshape{#5}%
  \selectfont}%
\fi\endgroup%
\begin{picture}(3703,1369)(2885,-7397)
\put(4051,-6886){\makebox(0,0)[lb]{\smash{{\SetFigFont{9}{10.8}{\rmdefault}{\mddefault}{\updefault}{\color[rgb]{0,0,0}$\mathcal{M}$}%
}}}}
\put(3301,-7196){\rotatebox{45.0}{\makebox(0,0)[lb]{\smash{{\SetFigFont{9}{10.8}{\rmdefault}{\mddefault}{\updefault}{\color[rgb]{0,0,0}$\mathcal{I}^+$}%
}}}}}
\put(5881,-6926){\rotatebox{315.0}{\makebox(0,0)[lb]{\smash{{\SetFigFont{9}{10.8}{\rmdefault}{\mddefault}{\updefault}{\color[rgb]{0,0,0}$\mathcal{I}^+$}%
}}}}}
\put(3843,-6654){\rotatebox{45.0}{\makebox(0,0)[lb]{\smash{{\SetFigFont{9}{10.8}{\rmdefault}{\mddefault}{\updefault}{\color[rgb]{0,0,0}$\mathcal{CH}^+$}%
}}}}}
\put(5385,-6426){\rotatebox{315.0}{\makebox(0,0)[lb]{\smash{{\SetFigFont{9}{10.8}{\rmdefault}{\mddefault}{\updefault}{\color[rgb]{0,0,0}$\mathcal{CH}^+$}%
}}}}}
\put(4761,-7276){\makebox(0,0)[lb]{\smash{{\SetFigFont{9}{10.8}{\rmdefault}{\mddefault}{\updefault}{\color[rgb]{0,0,0}$\Sigma$}%
}}}}
\put(5867,-6576){\makebox(0,0)[lb]{\smash{{\SetFigFont{9}{10.8}{\rmdefault}{\mddefault}{\updefault}{\color[rgb]{0,0,0}$i^+$}%
}}}}
\put(3453,-6576){\makebox(0,0)[lb]{\smash{{\SetFigFont{9}{10.8}{\rmdefault}{\mddefault}{\updefault}{\color[rgb]{0,0,0}$i^+$}%
}}}}
\end{picture}%

%% file: glob_pert.pstex_t
\begin{picture}(0,0)%
\includegraphics{glob_pert.pstex}%
\end{picture}%
\setlength{\unitlength}{2960sp}%
\begingroup\makeatletter\ifx\SetFigFont\undefined%
\gdef\SetFigFont#1#2#3#4#5{%
  \reset@font\fontsize{#1}{#2pt}%
  \fontfamily{#3}\fontseries{#4}\fontshape{#5}%
  \selectfont}%
\fi\endgroup%
\begin{picture}(3703,1874)(2885,-7397)
\put(4051,-6436){\rotatebox{45.0}{\makebox(0,0)[lb]{\smash{{\SetFigFont{9}{10.8}{\rmdefault}{\mddefault}{\updefault}{\color[rgb]{0,0,0}$\mathcal{CH}^+$}%
}}}}}
\put(5101,-6136){\rotatebox{315.0}{\makebox(0,0)[lb]{\smash{{\SetFigFont{9}{10.8}{\rmdefault}{\mddefault}{\updefault}{\color[rgb]{0,0,0}$\mathcal{CH}^+$}%
}}}}}
\put(4051,-6886){\makebox(0,0)[lb]{\smash{{\SetFigFont{9}{10.8}{\rmdefault}{\mddefault}{\updefault}{\color[rgb]{0,0,0}$\mathcal{M}$}%
}}}}
\put(3301,-7196){\rotatebox{45.0}{\makebox(0,0)[lb]{\smash{{\SetFigFont{9}{10.8}{\rmdefault}{\mddefault}{\updefault}{\color[rgb]{0,0,0}$\mathcal{I}^+$}%
}}}}}
\put(5881,-6926){\rotatebox{315.0}{\makebox(0,0)[lb]{\smash{{\SetFigFont{9}{10.8}{\rmdefault}{\mddefault}{\updefault}{\color[rgb]{0,0,0}$\mathcal{I}^+$}%
}}}}}
\put(5538,-7128){\rotatebox{45.0}{\makebox(0,0)[lb]{\smash{{\SetFigFont{9}{10.8}{\rmdefault}{\mddefault}{\updefault}{\color[rgb]{0,0,0}$\mathcal{H}^+$}%
}}}}}
\put(3599,-6817){\rotatebox{315.0}{\makebox(0,0)[lb]{\smash{{\SetFigFont{9}{10.8}{\rmdefault}{\mddefault}{\updefault}{\color[rgb]{0,0,0}$\mathcal{H}^+$}%
}}}}}
\put(4761,-7276){\makebox(0,0)[lb]{\smash{{\SetFigFont{9}{10.8}{\rmdefault}{\mddefault}{\updefault}{\color[rgb]{0,0,0}$\Sigma_+$}%
}}}}
\end{picture}%

%% file: glob_form.pstex_t
\begin{picture}(0,0)%
\includegraphics{glob_form.pstex}%
\end{picture}%
\setlength{\unitlength}{2960sp}%
\begingroup\makeatletter\ifx\SetFigFont\undefined%
\gdef\SetFigFont#1#2#3#4#5{%
  \reset@font\fontsize{#1}{#2pt}%
  \fontfamily{#3}\fontseries{#4}\fontshape{#5}%
  \selectfont}%
\fi\endgroup%
\begin{picture}(4477,1532)(2733,-7397)
\put(4051,-6886){\makebox(0,0)[lb]{\smash{{\SetFigFont{9}{10.8}{\rmdefault}{\mddefault}{\updefault}{\color[rgb]{0,0,0}$\mathcal{M}$}%
}}}}
\put(3301,-7196){\rotatebox{45.0}{\makebox(0,0)[lb]{\smash{{\SetFigFont{9}{10.8}{\rmdefault}{\mddefault}{\updefault}{\color[rgb]{0,0,0}$\mathcal{I}^+_B$}%
}}}}}
\put(5881,-6926){\rotatebox{315.0}{\makebox(0,0)[lb]{\smash{{\SetFigFont{9}{10.8}{\rmdefault}{\mddefault}{\updefault}{\color[rgb]{0,0,0}$\mathcal{I}^+_A$}%
}}}}}
\put(4699,-5985){\makebox(0,0)[lb]{\smash{{\SetFigFont{9}{10.8}{\rmdefault}{\mddefault}{\updefault}{\color[rgb]{0,0,0}$\mathcal{S}$}%
}}}}
\put(3807,-6674){\rotatebox{45.0}{\makebox(0,0)[lb]{\smash{{\SetFigFont{9}{10.8}{\rmdefault}{\mddefault}{\updefault}{\color[rgb]{0,0,0}$\mathcal{CH}^+_B$}%
}}}}}
\put(3967,-6168){\rotatebox{45.0}{\makebox(0,0)[lb]{\smash{{\SetFigFont{9}{10.8}{\rmdefault}{\mddefault}{\updefault}{\color[rgb]{0,0,0}$\mathcal{S}_B$}%
}}}}}
\put(5386,-6069){\rotatebox{315.0}{\makebox(0,0)[lb]{\smash{{\SetFigFont{9}{10.8}{\rmdefault}{\mddefault}{\updefault}{\color[rgb]{0,0,0}$\mathcal{S}_A$}%
}}}}}
\put(5904,-6559){\makebox(0,0)[lb]{\smash{{\SetFigFont{9}{10.8}{\rmdefault}{\mddefault}{\updefault}{\color[rgb]{0,0,0}$i^+_A$}%
}}}}
\put(3402,-6539){\makebox(0,0)[lb]{\smash{{\SetFigFont{9}{10.8}{\rmdefault}{\mddefault}{\updefault}{\color[rgb]{0,0,0}$i^+_B$}%
}}}}
\put(2733,-7245){\makebox(0,0)[lb]{\smash{{\SetFigFont{9}{10.8}{\rmdefault}{\mddefault}{\updefault}{\color[rgb]{0,0,0}$i^0_B$}%
}}}}
\put(6564,-7260){\makebox(0,0)[lb]{\smash{{\SetFigFont{9}{10.8}{\rmdefault}{\mddefault}{\updefault}{\color[rgb]{0,0,0}$i^0_A$}%
}}}}
\put(5557,-7121){\rotatebox{45.0}{\makebox(0,0)[lb]{\smash{{\SetFigFont{9}{10.8}{\rmdefault}{\mddefault}{\updefault}{\color[rgb]{0,0,0}$\mathcal{H}^+_A$}%
}}}}}
\put(3652,-6881){\rotatebox{315.0}{\makebox(0,0)[lb]{\smash{{\SetFigFont{9}{10.8}{\rmdefault}{\mddefault}{\updefault}{\color[rgb]{0,0,0}$\mathcal{H}^+_B$}%
}}}}}
\put(5370,-6430){\rotatebox{315.0}{\makebox(0,0)[lb]{\smash{{\SetFigFont{9}{10.8}{\rmdefault}{\mddefault}{\updefault}{\color[rgb]{0,0,0}$\mathcal{CH}^+_A$}%
}}}}}
\put(4761,-7276){\makebox(0,0)[lb]{\smash{{\SetFigFont{9}{10.8}{\rmdefault}{\mddefault}{\updefault}{\color[rgb]{0,0,0}$\Sigma$}%
}}}}
\end{picture}%

%% file: glob_gam.pstex_t
\begin{picture}(0,0)%
\includegraphics{glob_gam.pstex}%
\end{picture}%
\setlength{\unitlength}{2960sp}%
\begingroup\makeatletter\ifx\SetFigFont\undefined%
\gdef\SetFigFont#1#2#3#4#5{%
  \reset@font\fontsize{#1}{#2pt}%
  \fontfamily{#3}\fontseries{#4}\fontshape{#5}%
  \selectfont}%
\fi\endgroup%
\begin{picture}(1594,991)(4994,-7397)
\put(5880,-6902){\rotatebox{315.0}{\makebox(0,0)[lb]{\smash{{\SetFigFont{9}{10.8}{\rmdefault}{\mddefault}{\updefault}{\color[rgb]{0,0,0}$\mathcal{I}^+$}%
}}}}}
\put(5551,-7111){\rotatebox{45.0}{\makebox(0,0)[lb]{\smash{{\SetFigFont{9}{10.8}{\rmdefault}{\mddefault}{\updefault}{\color[rgb]{0,0,0}$\mathcal{H}^+_A$}%
}}}}}
\put(5026,-6586){\rotatebox{315.0}{\makebox(0,0)[lb]{\smash{{\SetFigFont{9}{10.8}{\rmdefault}{\mddefault}{\updefault}{\color[rgb]{0,0,0}$v=V_0$}%
}}}}}
\put(5401,-6661){\makebox(0,0)[lb]{\smash{{\SetFigFont{9}{10.8}{\rmdefault}{\mddefault}{\updefault}{\color[rgb]{0,0,0}$\gamma_A$}%
}}}}
\put(5842,-6526){\makebox(0,0)[lb]{\smash{{\SetFigFont{9}{10.8}{\rmdefault}{\mddefault}{\updefault}{\color[rgb]{0,0,0}$i^+_A$}%
}}}}
\end{picture}%

%% file: glob_boot.pstex_t
\begin{picture}(0,0)%
\includegraphics{glob_boot.pstex}%
\end{picture}%
\setlength{\unitlength}{2960sp}%
\begingroup\makeatletter\ifx\SetFigFont\undefined%
\gdef\SetFigFont#1#2#3#4#5{%
  \reset@font\fontsize{#1}{#2pt}%
  \fontfamily{#3}\fontseries{#4}\fontshape{#5}%
  \selectfont}%
\fi\endgroup%
\begin{picture}(3703,1873)(2885,-7397)
\put(3301,-7196){\rotatebox{45.0}{\makebox(0,0)[lb]{\smash{{\SetFigFont{9}{10.8}{\rmdefault}{\mddefault}{\updefault}{\color[rgb]{0,0,0}$\mathcal{I}^+$}%
}}}}}
\put(3676,-6511){\rotatebox{45.0}{\makebox(0,0)[lb]{\smash{{\SetFigFont{9}{10.8}{\rmdefault}{\mddefault}{\updefault}{\color[rgb]{0,0,0}$\mathcal{CH}^+_B$}%
}}}}}
\put(5926,-6886){\rotatebox{315.0}{\makebox(0,0)[lb]{\smash{{\SetFigFont{9}{10.8}{\rmdefault}{\mddefault}{\updefault}{\color[rgb]{0,0,0}$\mathcal{I}^+$}%
}}}}}
\put(5562,-6190){\rotatebox{315.0}{\makebox(0,0)[lb]{\smash{{\SetFigFont{9}{10.8}{\rmdefault}{\mddefault}{\updefault}{\color[rgb]{0,0,0}$\mathcal{CH}^+_A$}%
}}}}}
\put(3676,-6886){\rotatebox{315.0}{\makebox(0,0)[lb]{\smash{{\SetFigFont{9}{10.8}{\rmdefault}{\mddefault}{\updefault}{\color[rgb]{0,0,0}$\mathcal{H}^+_B$}%
}}}}}
\put(5551,-7111){\rotatebox{45.0}{\makebox(0,0)[lb]{\smash{{\SetFigFont{9}{10.8}{\rmdefault}{\mddefault}{\updefault}{\color[rgb]{0,0,0}$\mathcal{H}^+_A$}%
}}}}}
\put(4351,-6661){\rotatebox{45.0}{\makebox(0,0)[lb]{\smash{{\SetFigFont{9}{10.8}{\rmdefault}{\mddefault}{\updefault}{\color[rgb]{0,0,0}$u=U$}%
}}}}}
\put(4726,-6286){\rotatebox{315.0}{\makebox(0,0)[lb]{\smash{{\SetFigFont{9}{10.8}{\rmdefault}{\mddefault}{\updefault}{\color[rgb]{0,0,0}$v=V$}%
}}}}}
\put(4726,-7186){\makebox(0,0)[lb]{\smash{{\SetFigFont{9}{10.8}{\rmdefault}{\mddefault}{\updefault}{\color[rgb]{0,0,0}$\Sigma$}%
}}}}
\put(5401,-6661){\makebox(0,0)[lb]{\smash{{\SetFigFont{9}{10.8}{\rmdefault}{\mddefault}{\updefault}{\color[rgb]{0,0,0}$\gamma_A$}%
}}}}
\put(3826,-6661){\makebox(0,0)[lb]{\smash{{\SetFigFont{9}{10.8}{\rmdefault}{\mddefault}{\updefault}{\color[rgb]{0,0,0}$\gamma_B$}%
}}}}
\end{picture}%

%% file: glob_cosmo.pstex_t
\begin{picture}(0,0)%
\includegraphics{glob_cosmo.pstex}%
\end{picture}%
\setlength{\unitlength}{2960sp}%
\begingroup\makeatletter\ifx\SetFigFont\undefined%
\gdef\SetFigFont#1#2#3#4#5{%
  \reset@font\fontsize{#1}{#2pt}%
  \fontfamily{#3}\fontseries{#4}\fontshape{#5}%
  \selectfont}%
\fi\endgroup%
\begin{picture}(7516,1936)(2843,-7460)
\put(4051,-6436){\rotatebox{45.0}{\makebox(0,0)[lb]{\smash{{\SetFigFont{9}{10.8}{\rmdefault}{\mddefault}{\updefault}{\color[rgb]{0,0,0}$\mathcal{CH}^+$}%
}}}}}
\put(5101,-6136){\rotatebox{315.0}{\makebox(0,0)[lb]{\smash{{\SetFigFont{9}{10.8}{\rmdefault}{\mddefault}{\updefault}{\color[rgb]{0,0,0}$\mathcal{CH}^+$}%
}}}}}
\put(4051,-6886){\makebox(0,0)[lb]{\smash{{\SetFigFont{9}{10.8}{\rmdefault}{\mddefault}{\updefault}{\color[rgb]{0,0,0}$\mathcal{M}$}%
}}}}
\put(7651,-6511){\rotatebox{45.0}{\makebox(0,0)[lb]{\smash{{\SetFigFont{9}{10.8}{\rmdefault}{\mddefault}{\updefault}{\color[rgb]{0,0,0}$\mathcal{CH}^+$}%
}}}}}
\put(8701,-6061){\rotatebox{315.0}{\makebox(0,0)[lb]{\smash{{\SetFigFont{9}{10.8}{\rmdefault}{\mddefault}{\updefault}{\color[rgb]{0,0,0}$\mathcal{CH}^+$}%
}}}}}
\put(5551,-7411){\makebox(0,0)[lb]{\smash{{\SetFigFont{9}{10.8}{\rmdefault}{\mddefault}{\updefault}{\color[rgb]{0,0,0}$\Sigma$}%
}}}}
\put(6151,-6436){\makebox(0,0)[lb]{\smash{{\SetFigFont{9}{10.8}{\rmdefault}{\mddefault}{\updefault}{\color[rgb]{0,0,0}$r=\infty$}%
}}}}
\put(9601,-6586){\makebox(0,0)[lb]{\smash{{\SetFigFont{9}{10.8}{\rmdefault}{\mddefault}{\updefault}{\color[rgb]{0,0,0}$r=\infty$}%
}}}}
\end{picture}%